\documentclass[a4paper,12pt,reqno]{amsart} 
\usepackage[lmargin=3.7cm,rmargin=2.7cm,bmargin=3.5cm,tmargin=3cm]{geometry}
\usepackage[numbers]{natbib}

\usepackage{./symbols}

\begin{document}

\title[Importance sampling versus standard averages]{
Importance sampling correction versus standard averages of reversible MCMCs in terms of the asymptotic variance}
\author{Jordan Franks}
\address{School of Mathematics, Statistics and Physics, Newcastle University, NE1 7RU Newcastle, United Kingdom}
\email{franks@iki.fi}

\author{Matti Vihola}
\address{Department of Mathematics and Statistics, P.O.Box 35, FI-40014 University of Jyväskylä, Finland}
\email{matti.vihola@iki.fi}


\begin{abstract}
  We establish an ordering criterion for the asymptotic variances of two consistent Markov chain Monte Carlo (MCMC) estimators: an importance sampling (IS) estimator, based on an approximate reversible chain and subsequent IS weighting, and a standard MCMC estimator, based on an exact reversible chain.
  Essentially, we relax the criterion of the Peskun type covariance ordering by considering two different invariant probabilities, and obtain, in place of a strict ordering of asymptotic variances, a bound of the asymptotic variance of IS by that of the direct MCMC.
Simple examples show that IS can have arbitrarily better or worse asymptotic variance than Metropolis-Hastings and delayed-acceptance (DA) MCMC.  Our ordering implies that IS is guaranteed to be competitive up to a factor depending on the supremum of the (marginal) IS weight. 
  We elaborate upon the criterion in case of unbiased estimators as part of an auxiliary variable framework. 
  We show how the criterion implies asymptotic variance guarantees for IS in terms of pseudo-marginal (PM) and DA corrections, essentially if the ratio of exact and approximate likelihoods is bounded.  
  We also show that convergence of the IS chain can be less affected by unbounded high-variance unbiased estimators than PM and DA chains.
\end{abstract}

\maketitle
\section{Introduction}\label{sec:intro}

Let $\nu(\theta, z)\ud \theta \ud z$  be a probability measure on a jointly measurable space $\mathbf{T}\times \mathbf{Z}$ with $\sigma$-finite dominating measure $\ud \theta \ud z$, and suppose one desires to calculate expectations with respect to $\nu$ or its marginal
$
\nu^*(\theta) \defeq \int \nu(\theta, z)\ud z.
$
In many scenarios of interest, $\nu^*(\theta)$ is intractable to evaluate and $\mathbf{Z}$ is high-dimensional.
Even if the full joint density $\nu(\theta,z)$ is tractable up to a normalising constant, high-dimensional Markov chain Monte Carlo (MCMC) based on targeting $\nu$ with Metropolis-Hastings (MH) is often inefficient or even useless due to difficulty with the design of proposal distribution \cite{liuWK}.  

To deal with these issues, often one can transform the high-dimensional MCMC into a pseudo-marginal (PM) \cite{andrieu-roberts} or approximate marginal MCMC \cite{vihola-helske-franks}.  Then the proposal distribution can live on the low-dimensional space $\mathbf{T}$, and the resulting chains are often much more efficient.  The PM approach is asymptotically exact, while the approximate marginal approach requires an importance sampling (IS) correction to make it so.  We next describe these approaches, and give our main result comparing the relative efficiency of these approaches in terms of the asymptotic variance.

\subsection{Pseudo-marginal Markov chain Monte Carlo}
The PM approach is based on replacing $\nu^*(\theta)$ with a non-negative unbiased estimator $\hat{\nu}^*(\theta)$ of $\nu^*(\theta)$ (up to constant) within the standard (but assumed unavailable) $\MH$ algorithm targeting $\nu^*$ \cite{andrieu-roberts,lin-liu-sloan}.  That is, we assume there is a constant $c_\nu>0$ such that
\begin{equation}\label{eq:unbiased-marginal}
  \E[\hat{\nu}^*(\theta)] = c_\nu\,\nu^*(\theta)
  \end{equation}
for all $\theta$.  A standard example of such an estimator is
\begin{equation*}
\hat{\nu}^*(\theta) = \frac{1}{m} \sum_{i=1}^m \frac{ \nu(\theta, Z_i)}{Q_\theta(Z_i)},
\end{equation*}
where $Z_i$ are sampled i.i.d.~from some instrumental distribution $Q_\theta(\cdot)$ satisfying $Q_\theta(z) = 0 \implies \nu(\theta, z)=0$.    

An unbiased estimator satisfying \eqref{eq:unbiased-marginal} will allow for calculation of marginal expectations with respect to $\nu^*$ using the approaches we consider, but often one can do much better, allowing also for joint expectations with respect to $\nu$ \cite{andrieu-doucet-holenstein,vihola-helske-franks}.
That is, suppose one has access to non-negative unbiased estimators for a subclass $\mathcal{L}^1(\nu)$ of functions which we now describe.

With $\theta\in \mathbf{T}$, $\est^{(i)}\in [0,\infty)$ and $z^{(i)}\in \mathbf{Z}$ for $i=1,\ldots, m$, set
  $$
 \est(g) \defeq \frac{1}{m} \sum_{i=1}^m \zeta^{(i)} g(\theta,z^{(i)})
 $$
for $g\in L^1(\nu)$.
Let $Q$ be a probability kernel from $\mathbf{T}$ to $\mathbf{V}\defeq [0,\infty)^m \times \mathbf{Z}^m$.
Let $\mathcal{L}^1(\nu)$ consist of those functions $f\in L^1(\nu)$ such that for $g\in \{f,|f|\}$,
\begin{equation}\label{eq:unbiased}
\int Q_\theta(\ud v) \est(g) = c_\nu\, \nu^*(\theta)\int g(\theta, z) \nu(\ud z| \theta)
\end{equation}
for all $\theta \in \mathbf{T}$ and $v\defeq(\zeta^{(1:m)}, z^{(1:m)})\in\mathbf{V}$, where $c_\nu>0$ is some fixed constant not depending on $(\theta,v)$, and $\nu(\ud z| \theta)$ denotes a regular conditional probability of $\nu$ given $\theta$.

Note that if $1\in \mathcal{L}^1(\nu)$, then \eqref{eq:unbiased-marginal} is satisfied with $\hat{\nu}^*(\theta) \defeq \frac{1}{m}\sum_{i=1}^m \zeta^{(i)}$, and $L^1(\nu^*)$ is naturally included into $\mathcal{L}^1(\nu)$ via $f(\theta,z) \defeq f(\theta)$ for $f\in L^1(\nu^*)$.  In the following, we will always assume that $1\in \mathcal{L}^1(\nu)$, since the schemes we consider require this for consistency.  

Let $q$ be a transition density on $\mathbf{T}$, where $q_\theta(\theta')$ denotes the probability to move from $\theta$ to $\theta'$.  With $(\Theta_0,\zeta_0^{(1:m)}, Z_0^{(1:m)})\in \mathbf{T}\times[0,\infty)^m\times \mathbf{Z}^m$ initial values with $\sum_{i=1}^m \zeta_0^{(i)}>0$, for $k=1,\ldots, n$, the PM iteration is given in Algorithm \ref{alg:pm} \citep[see][]{andrieu-roberts}.
  \begin{algorithm}[t]\caption{Pseudo-marginal algorithm, for iteration $k\ge 1$.}
    \label{alg:pm}
  \begin{enumerate}[(PM 1)]
\item Propose a transition $\Theta' \sim q_{\Theta_{k-1}}(\cdot)$.
\item Given $\Theta'$, generate $(\zeta'^{(1:m)}, Z'^{(1:m)})$, and with probability
  $$
  \min\bigg\{1, \frac{ q_{\Theta'}(\Theta_{k-1}) \sum_{i=1}^m \zeta'^{(i)}}{q_{\Theta_{k-1}}(\Theta') \sum_{i=1}^m \zeta_{k-1}^{(i)}} \bigg\}
  $$
  set $(\Theta_k, \zeta_k^{(i)}, Z_k^{(i)}) \leftarrow (\Theta', \zeta'^{(i)},Z'^{(i)})$.  Otherwise, set
  $(\Theta_k, \zeta_k^{(i)}, Z_k^{(i)})) \leftarrow (\Theta_{k-1}, \zeta_{k-1}^{(i)}, Z_{k-1}^{(i)})$.
  \end{enumerate}
\end{algorithm}

Assuming $1\in\mathcal{L}^1(\nu)$ and the PM chain is Harris ergodic (see Section \ref{sec:preliminaries}), for $f\in \mathcal{L}^1(\nu)$ the estimator
\begin{equation}\label{eq:direct-estimator}
\frac{1}{n} \sum_{k=1}^n \hat{\zeta}_k^{(i)} f(\Theta_k, Z_k^{(i)})
\xrightarrow{n\to\infty} \nu(f)\defeq \E_{\nu}[f]
\end{equation}
with $\hat{\zeta}_k^{(i)}\defeq \zeta_k^{(i)} / \sum_{j=1}^m \zeta_k^{(j)}$, 
is a consistent estimator for $\nu(f)$ \cite{andrieu-roberts}. 

\subsection{Accelerations based on an approximation}
Suppose one has an approximation $\mu^*$ of $\nu^*$, by which we mean that $\mu^*$ is some probability measure on $\mathbf{T}$ such that
\begin{equation}\label{eq:support-cond}
\mu^*(\theta)=0\implies \nu^*(\theta)=0,
\end{equation}
and there is some constant $c_\mu>0$ (perhaps unknown) such that we can evaluate unnormalised $\mu_u^*(\theta) = c_\mu \mu^*(\theta)$ for all $\theta\in \mathbf{T}$.  The approximation $\mu^*$ could arise, for example, when subsampling data \cite{banterle-grazian-lee-robert, quiroz-tran-villani-kohn} or using a more tractable diffusion model instead of a Markov jump process model \cite{golightly-henderson-sherlock}, giving rise to the approximate posterior $\mu^*$.  We next describe delayed-acceptance (DA) in two variants, and IS, all of which make use of the approximation $\mu^*$.  
\subsubsection{Delayed-acceptance MCMC in two variants}
If one has an approximation $\mu^*$ of $\nu^*$ as above, one can use a
PM acceleration technique
known as DA \cite{christen-fox,lin-liu-sloan,liu-mc}, which has garnered considerable interest.  
With $(\Theta_0,\est_0^{(1:m)}, Z_0^{(1:m)}) \in \mathbf{T}\times [0,\infty)^m\times \mathbf{Z}^m$ initial values with $\mu^*(\Theta_0)>0$ and $\sum_{i=1}^m \est_0^{(i)} >0$, for $k=1,\ldots, n$, iterate as given in Algorithm \ref{alg:dapr-mh}.
\begin{algorithm}[t]
\caption{Delayed-acceptance $(\DApr)$, for iteration $k\ge 1$.}
\label{alg:dapr-mh}
\begin{enumerate}[($\DApr$ 1)]
\item Propose a transition $\Theta' \sim q_{\Theta_{k-1}}(\cdot)$.
\item \label{alg:dapr-first} Proceed to Step ($\DApr$ \ref{alg:dapr-second}) with probability
  $$
\min\bigg\{1, \frac{ \mu_u^* (\Theta') q_{\Theta'}(\Theta_{k-1})}{\mu_u^*(\Theta_{k-1}) q_{\Theta_{k-1}}(\Theta')} \bigg\}.
$$
Otherwise, set $(\Theta_k, \est_k^{(i)}, Z_k^{(i)})\leftarrow (\Theta_{k-1}, \est_{k-1}^{(i)}, Z_{k-1}^{(i)}))$ and exit.

\item \label{alg:dapr-second}
  Given $\Theta'$, generate $(\zeta'^{(1:m)}, Z'^{(1:m)})$.
  With probability
  $$
  \min\bigg\{1, \frac{(\sum_{i=1}^m \est'^{(i)})/ \mu_u^*(\Theta')}{(\sum_{i=1}^m \est_{k-1}^{(i)})/ \mu_u^*(\Theta_{k-1})} \bigg\}
$$
set $(\Theta_k, \est_k^{(i)}, Z_k^{(i)})\leftarrow (\Theta', \est'^{(i)}, Z'^{(i)})$.  Otherwise, set
$
(\Theta_k, \est_k^{(i)}, Z_k^{(i)}) \leftarrow (\Theta_{k-1}, \est_{k-1}^{(i)}, Z_{k-1}^{(i)}).
$
\end{enumerate}
\end{algorithm}

The DA estimator for $\nu(f)$ is given in \eqref{eq:direct-estimator}, which is the same as in the PM case.
Note that Step ($\DApr$ \ref{alg:dapr-second}), which involves possibly expensive unbiased estimator generation, is only run if the proposal $\Theta'$ is assigned sufficient approximate probability in Step ($\DApr$ \ref{alg:dapr-first}) which means it is likely to be accepted in Step ($\DApr$ \ref{alg:dapr-second}).

Consider now Algorithm \ref{alg:dag-mh}, which is a variant, $\DAg$, of $\DApr$ Algorithm \ref{alg:dapr-mh}.  $\DAg$ is considered in \citep[][`surrogate transition method,' Section 9.4.3]{liu-mc}.
\begin{algorithm}
  \caption{Delayed-acceptance $(\DAg)$, for iteration $k\ge 1$.}
  \label{alg:dag-mh}
\begin{enumerate}[($\DAg$ 1)]
\item Propose a transition $\Theta' \sim q_{\Theta_{k-1}}(\cdot)$.
\item \label{alg:dag-first} With probability
  $$
\min\bigg\{1, \frac{ \mu_u^* (\Theta') q_{\Theta'}(\Theta_{k-1})}{\mu_u^*(\Theta_{k-1}) q_{\Theta_{k-1}}(\Theta')} \bigg\}
$$
set $\Theta''\leftarrow \Theta'$.  Otherwise, set $\Theta''\leftarrow \Theta_{k-1}$.

\item \label{alg:dag-second}
Given $\Theta''$, generate $(\est''^{(1:m)}, Z''^{(1:m)})$,
With probability
  $$
\min\bigg\{1, \frac{(\sum_{i=1}^m \est''^{(i)})/ \mu_u^*(\Theta'')}{(\sum_{i=1}^m \est_{k-1}^{(i)})/ \mu_u^*(\Theta_{k-1})} \bigg\}
$$
set $(\Theta_k, \est_k^{(i)}, Z_k^{(i)}) \leftarrow (\Theta'', \est''^{(i)}, Z''^{(i)})$.

Otherwise, set
$
(\Theta_k, \est_k^{(i)}, Z_k^{(i)}) \leftarrow (\Theta_{k-1}, \est_{k-1}^{(i)}, Z_{k-1}^{(i)}).
$
\end{enumerate}
\end{algorithm}
In the deterministic case $\frac{1}{m} \sum_{i=1}^m \est^{(i)}= c_\nu \nu^*(\theta)$ almost surely for all $\theta$ with $(\est^{(1:m)}, Z^{(1:m)})\sim Q_{\theta}(\cdot)$, then $\DApr$ and $\DAg$ have the same transition kernels (Proposition \ref{prop:das}).  But in general $\DAg$ has lower asymptotic variance than $\DApr$ (Proposition \ref{prop:das}), although $\DApr$ is probably more computationally efficient.  This is evident from the fact that Step ($\DAg$ \ref{alg:dag-second}) is performed at every iteration of $\DAg$ (Algorithm \ref{alg:dag-mh}).

\subsubsection{Importance sampling correction of approximate MCMC}
MCMC-IS (Algorithm \ref{alg:is}) consists of targeting an approximation $\mu^*$ of $\nu^*$ 
with MCMC, and then using importance sampling (IS) correction over the latent states \cite{doss,gilks-roberts,glynn-iglehart,hastings,parpas-ustun-webster-tran,vihola-helske-franks}.  
Let $\mu^*$ be an approximation of $\nu^*$ as in \eqref{eq:support-cond}.
\begin{algorithm}
  \caption{Importance sampling correction of approximate MCMC}
  \label{alg:is}
\begin{enumerate}[({IS Phase} 1)]
\item\label{alg:is-first} Let $(\Theta_0,\est_0^{(1:m)}, Z_0^{(1:m)})\in \mathbf{X}\times(0,\infty)^m\times \mathbf{Z}^m$ be some initial values with $\mu^*(\Theta_0)>0$ and $\sum_{i=1}^m \est_0^{(i)}>0$.  For $k=1,\ldots, n$, do:
  \begin{enumerate}[i.]
  \item Propose a transition $\Theta'\sim q_{\Theta_{k-1}}(\cdot)$.
  \item 
  With probability
  $$
\min\bigg\{1, \frac{\mu_u^*(\Theta') q_{\Theta'}(\Theta_{k-1})}{\mu_u^*(\Theta_{k-1})q_{\Theta_{k-1}}(\Theta')} \bigg\}
$$
set $\Theta_k \leftarrow \Theta'$.  Otherwise, set $\Theta_k\leftarrow \Theta_{k-1}$.
  \end{enumerate}
  \item\label{alg:is-second} For each $k\in \{1,\ldots, n\}$, given $\Theta_k$, generate $(\zeta_k^{(1:m)}, Z_k^{(1:m)})$.
  \newline With $\xi_k^{(i)}\defeq \est_k^{(i)} / \mu_u^*(\theta)$, form the IS estimator
    \begin{equation}\label{eq:indirect-estimator}
      E_n^{\IS}(f)\defeq \frac{\sum_{k=1}^n \sum_{i=1}^m \xi_k^{(i)} f(\Theta_k,Z_k^{(i)})}{\sum_{k=1}^n \sum_{i=1}^m\xi_k^{(i)} }
      \xrightarrow{n\to\infty} \nu(f),
      \end{equation}
consistent if the Phase \ref{alg:is-first} chain is Harris ergodic and $1, f\in \mathcal{L}^1(\nu)$.
\end{enumerate}
\end{algorithm}
Note that IS Phase \ref{alg:is-second}, which involves the generation of unbiased estimators, may be done independently for each $k$ which allows for efficient parallelisation.

\subsection{Defining the asymptotic variance}
As PM/DA and MCMC-IS are viable approaches for consistent inference, the central question is which one should be used.
The standard measure of statistical efficiency for MCMCs is the \emph{asymptotic variance}.

\begin{definition}[Asymptotic variance]\label{def:asvar}
  Let $(X_k)$ be a $\nu$-Harris ergodic Markov chain with transition $L$.  For $f\in L^2(\nu)$ the \emph{asymptotic variance} of $f$ with respect to $L$ is defined, whenever the limit exists in $[0,\infty]$, as
  \begin{equation}\label{def:asvar-eq}
\Var(L, f) \defeq \lim_{n \ra \infty} \E\Big[\Big(\frac{1}{\sqrt{n}} \sum_{k=1}^n [f(X_k^{(s)}) - \nu(f)] \Big)^2\Big],
\end{equation}
where $(X^{(s)}_k)$ denotes a stationary version of the chain $(X_k)$, i.e.~$X^{(s)}_0 \sim \nu$.  
  \end{definition}
For reversible $L$, which is the focus of this paper, $\Var(L,f)$ always exists in $[0,\infty]$ \citep[see][]{tierney-note}.  Moreover, a CLT holds under general conditions.   
\begin{proposition}\label{prop:kipnis-varadhan}
Let $(X_k)_{k\ge 1}$ be an aperiodic $\nu$-reversible Harris ergodic Markov chain with transition $L$.  If $f \in L^2(\nu)$ and $\Var(L, f)<\infty$, then, for all initial distributions,
\begin{equation}\label{eq:kipnis-varadhan}
  \frac{1}{\sqrt{n}}\Big(\sum_{k=1}^n [f(X_k) - \nu(f) ]\Big)
\xrightarrow{n\ra\infty}
\mathcal{N}\big(0, \Var(L, f)\big),
\qquad\text{in distribution},
\end{equation}
where $\mathcal{N}(a,b^2)$ is a normal distribution with mean $a$ and variance $b^2$.
\end{proposition}
Proposition \ref{prop:kipnis-varadhan} follows from \citep[Cor.~1.5]{kipnis-varadhan}, where it holds under all initial conditions because of the Harris ergodicity assumption \citep[see][Cor.~21.1.6]{doucMPS}.
Proposition \ref{prop:kipnis-varadhan} above explains the importance of the asymptotic variance, since it is the CLT limiting variance.  The asymptotic variance characterises the \emph{statistical efficiency} of the method in the asymptotic regime, but also characterises the finite sample efficiency in the finite regime \citep[see][]{rudolf-error}.

\subsection{Comparing the asymptotic variances}
We first define some objects.
Given $X_k = (\Theta_k, \est_k^{(1:m)}, Z_k^{(1:m)})$ and $f:\mathbf{X}\rightarrow \R$, define
$$
\est_k(f) = \frac{1}{m} \sum_{i=1}^m \est_k^{(i)} f(\Theta_k, Z_k^{(i)}),
\qquad
\hat{\est}_k^{(i)} \defeq \frac{ \est_k(f) }{ \est_k(1)}
\qquad
\xi_k(f) \defeq \frac{\est_k(f)}{\mu_u^*(\Theta_k)}.
$$
We also define the MCMC-IS kernel to be
\begin{equation}\label{eq:is-aug}
\bar{K}_{\theta v}(\ud \theta', \ud v') \defeq K_\theta(\ud \theta') Q_{\theta'}(\ud v')
\end{equation}
where $K$ is the approximate marginal $\MH$ kernel in IS (Algorithm \ref{alg:is}) Phase \eqref{alg:is-first}  with invariant measure $\mu^*$ \cite{vihola-helske-franks}.  $\bar{K}$ is $\bar{\mu}$-reversible, where
$
\bar{\mu}(\ud \theta, \ud v)\defeq \mu^*(\ud \theta) Q_{\theta}(\ud v).
$
Define the extended IS weights $w_u(X_k) \defeq \xi_k(1)$ and $w(X_k) \defeq w_u(X_k)*(c_\mu/ c_\nu)$.  

Assume $1\in \mathcal{L}^1(\nu)$ and \eqref{eq:support-cond} holds.  By our discussion of the asymptotic variance, $\Var\big(L,\hat{\zeta}(f)\big)$  is assigned to the PM/DA estimator $n^{-1}\sum_{k=1}^n \hat{\zeta}_k(f)$ given in \eqref{eq:direct-estimator}.  
Now let $\bar{K}$ be the MCMC-IS kernel defined in \eqref{eq:is-aug}, and note that the IS estimator \eqref{eq:indirect-estimator} can be written as
\begin{equation}\label{eq:slutsky}
E_n^{\IS}(f) =
\frac{\frac{1}{n}\sum_{k=1}^n \xi_k(f)}{\frac{1}{n} \sum_{k=1}^n \xi_k(1)}
=
\frac{\frac{1}{n}\sum_{k=1}^n w_u(X_k) \hat{\est}_k(f)}{\frac{1}{n} \sum_{k=1}^n w_u(X_k)}.
\end{equation}
Since $\Var(\bar{K}, w_u \hat{\est}(f) )$ is assigned to the numerator from the definition of the asymptotic variance, and the denominator converges almost surely to $c_\nu/c_\mu$ under a Harris ergodicity assumption, the asymptotic variance $\Var(\bar{K}, w \hat{\est}(f) )$ is assigned to the IS estimator by \eqref{eq:kipnis-varadhan} and Slutsky's lemma.

For a function $g:\mathbf{X}\rightarrow \mathbb{R}$ and probability $\nu$ on $\mathbf{X}$, define the norm
\begin{equation}\label{eq:esssup}
\norm{g}_{L^\infty(\nu)} \defeq \text{$\nu$-$\esssup_{x\in\mathbf{X}} |g(x)|$}.
\end{equation}
Let us define the marginal weight $w^*(\theta) \defeq \nu^*(\theta)/ \mu^*(\theta)$ and note that $\norm{w^*}_{L^\infty(\mu^*)} \le \norm{w}_{L^\infty(\bar{\mu})}$.  Under Harris ergodicity, we remark that
a consistent upper bound estimator for $\norm{w}_{L^\infty(\bar{\mu})}$ is given by
\begin{equation}\label{eq:esssup-est}
\uc_{n}\defeq \bigg(\frac{1}{n-n_b} \sum_{k=n_b+1}^n \xi_k(1)\bigg)^{-1} \max_{n_b< k\le n} \xi_k(1),
\end{equation}
which is moreover a consistent estimator for $\norm{w}_{L^\infty(\bar{\mu})}$ as $n_b, n\to\infty$, where $n_b\ge 1$ denotes the burn-in of the chain.

Let $L$ be the transition corresponding to the PM, $\DApr$, or $\DAg$ chain.  Then $L$ has invariant probability
$$
\pi(\ud \theta, \ud \est^{(1:m)}, \ud z^{(1:m)})
\defeq
c_\nu^{-1} \ud \theta Q_{\theta}(\ud \est^{(1:m)}, \ud z^{(1:m)}) \est(1)
$$
Define $\mathcal{L}^2(\nu) \defeq \{ f\in \mathcal{L}^1(\nu): f^2\in \mathcal{L}^1(\nu)\}$, where we recall $\mathcal{L}^1(\nu)$ was defined through \eqref{eq:unbiased}.

Our main result (Theorem \ref{thm:peskun-is}) in the present context says the following.
\begin{corollary}\label{cor:intro}
  Suppose $1\in \mathcal{L}^1(\nu)$.  Let $L$ be the transition kernel of the $\PM$, $\DApr$ or $\DAg$ chains defined in Algorithm \ref{alg:pm}-\ref{alg:dag-mh} respectively.   Suppose \eqref{eq:support-cond} holds, and let $\bar{K}$ be the MCMC-IS kernel \eqref{eq:is-aug} corresponding to Algorithm \ref{alg:is}, and let $f\in \mathcal{L}^2(\nu)$.  Suppose $K$ and $L$ are Harris ergodic and $\Var(\bar{K}, w \hat{\est}(f) )<\infty$. Set $\bar{f}\defeq f- \nu(f)$.  The following hold:
\begin{enumerate}[(i)]  
\item\label{cor:intro-marginal} If 
  $\norm{w^*}_{L^\infty(\mu^*)}<\infty$
  then, 
  $$
\Var\big(\bar{K}, w \hat{\est}(f)\big) \le \norm{w^*}_{L^\infty(\mu^*)} \Big(\Var\big(L, \hat{\est}(f)\big) + \Var_{\pi}\big(\hat{\est}(f)\big)\Big) + 3\, \Var_{\bar{\mu}}\big(w \hat{\est}(\bar{f}) \big).
$$ 
\item\label{cor:intro-full} If $\norm{w}_{L^\infty(\bar{\mu})} < \infty$, then
$$
\Var\big(\bar{K}, w \hat{\est}(f)\big) +\Var_{\bar{\mu}}\big( w \hat{\est}(\bar{f})\big) \le \norm{w}_{L^\infty(\bar{\mu})} \Big(\Var\big(L, \hat{\est}(f) \big) + \Var_{\pi}\big( \hat{\est}(f)\big) \Big).
$$
\item\label{cor:intro-below}
If $w \hat{\est}(f)\in L^2(\bar{\mu})$, then 
$$
\Var\big(\bar{K}, w \hat{\est}(f)\big) +\Var_{\bar{\mu}}\big( w \hat{\est}(\bar{f})\big) \ge (\text{$\bar{\mu}$-$\essinf w$}  )\Big(\Var\big(L, \hat{\est}(f) \big) + \Var_{\pi}\big( \hat{\est}(f)\big) \Big).
$$
\end{enumerate}
\end{corollary}

Although these bounds do not provide an ordering of asymptotic variances as in the Peskun-Tierney ordering for direct MCMCs, we could never hope these bounds to do so in our context: we give simple examples showing that PM/DA (resp.~IS) can do arbitrarily better than IS (resp.~PM/DA) in terms of the asymptotic variance in Appendix \ref{app:examples}.  IS seems to perform better than PM/DA when the approximation is good in the sense that the weight $w$ has a low supremum, as Corollary \ref{cor:intro} would suggest.  In the usual unbounded space case, this means that $\mu(\theta)$ would need to have fatter tails than $\dot{\nu}(\theta)$, or at least that $\nu(\theta)/\mu(\theta)$ is bounded, which can often be done by inflating $\mu_u(\theta)$ uniformly by a positive constant \citep[see][]{vihola-helske-franks}.

The rest of this paper is concerned with proving Corollary \ref{cor:intro} and other versions, for example, for general reversible chains, IS jump chains, and when $\mu_u^*(\theta)$ requires unbiased estimators.

\subsection{Previous work}
In various settings and different ways, we are not the first to compare direct MCMC with IS MCMC.
A study of self-normalised IS versus the independence $\MH$ has been made in \cite{liu}.  Asymptotic variances are explicitly computed and compared in some discrete examples in \cite{bassetti-diaconis} who find that IS and $\MH$ can be competitive, but that $\MH$ can do much better (see also \citep[Sect.~4.2]{bardsley-solonen-haario-laine}).  On the other hand, \cite{tran-scharth-pitt-kohn} study independent IS with unbiased estimators, and find that this performs better than PM in their experiments (see also \cite{chopin-jacob-papaspiliopoulos}).
The IS versus DA question is noted in \citep[Sect.~3.3.3]{cui-marzouk-willcox-scalable}, who mention the likely improvement of IS over DA in massive parallelisation.
A methodological comparison of the alternatives in the general MCMC and joint inference context is made in \cite{vihola-helske-franks}, who investigate empirically the relative efficiencies, finding that IS and DA can be competitive, with IS doing slightly better than DA in their experiments, with little or no parallelisation.  The gap widens with increased parallelisation, a known strength of the IS correction \citep[see][]{cui-marzouk-willcox-scalable,lee-yau-giles-doucet-holmes, vihola-helske-franks}.  

We consider here general reversible Markov chains, in particular PM/DA, and seek a Peskun type ordering of the asymptotic variances.

\subsection{Outline}\label{sec:main-outline}
After preliminaries in Section \ref{sec:preliminaries}, we state in Section \ref{sec:normalised} the Peskun type ordering result for normalised IS (Theorems \ref{thm:peskun}) and augmented IS kernels (Theorem \ref{thm:peskun-augmented}).  We define jump chains and self-normalised importance sampling (SNIS) in Section \ref{sec:jump}, before proceeding to Section \ref{sec:schemes}, where we consider a general auxiliary variable framework which accommodates IS and PM type schemes that use unbiased estimators.  Specific PM type algorithms and kernels which we consider are given in Section \ref{sec:application}, and we compare them with IS (Theorem \ref{thm:comparison}).
We discuss some stability considerations in Section \ref{sec:discussion}.
Proofs of the Peskun type orderings are given in Appendix \ref{app:peskun}.  Dirichlet form bounds and proof of the main comparison application (Theorem \ref{thm:comparison}) are found in  Appendix \ref{app:dirichlet}. Appendix \ref{app:misc} mentions some properties of augmented chains. Appendix \ref{app:examples} contains the examples mentioned earlier.  

\section{Notation and definitions}\label{sec:preliminaries}
\subsection{Notation}  
The spaces we consider $\mathbf{X}$ are assumed equipped with a $\gs$-algebra, denoted $\mathcal{B}(\mathbf{X})$, and with a $\gs$-finite dominating measure, denoted `$\ud x$.'  Product spaces will be assumed equipped with their product $\gs$-algebras and corresponding product measures.  If $\mu$ is a probability density on $\mathbf{X}$, we denote the corresponding probability measure with the same symbol, so that $\mu(\ud x) = \mu(x) \ud x$.   

For $p\in [1,\infty)$, we denote by $L^p(\mu)$ the Banach space of equivalence classes of measurable $f:\mathbf{X}\ra \R$ satisfying $\norm{f}_p<\infty$ under the norm $\norm{f}_{L^p(\mu)}\defeq \{\int \abs{f(x)}^p \mu(\ud x)\}^{1/p}$.  We similarly define $L^\infty(\mu)$ under the norm $\norm{f}_{L^\infty(\mu)}$ as in \eqref{eq:esssup}.  We denote by $L^p_0(\mu)$ the subset of $L^p(\mu)$ with $\mu(f)=0$, where $\mu(f)\defeq \int f(x) \mu (\ud x)$.  For $f\in L^1(\mu)$ and $K_x(\ud x')$ a Markov kernel on $\mathbf{X}$, we define $\mu K(A) \defeq \int\mu(\ud x) K_x(A)$ for $A\in \mathcal{B}(\mathbf{X})$, $K f(x)\defeq \int K_x(\ud x') f(x')$, and inductively $K^n f(x)\defeq K^{n-1}(Kf)(x)$ for $n\ge 2$.  For $f,g\in L^2(\mu)$, we define $\inner{f}{g}_\mu \defeq \int f(x) g(x) \mu(\ud x)$,
  $
  \norm{f}_\mu \defeq (\inner{f}{f}_\mu)^{1/2},
  $
  and $\Var_\mu(f) \defeq \mu(f^2) - \mu(f)^2$.  

  For $m\in \mathbb{N}$ and $x^{(i)}\in \mathbf{X}$ for $i=1,\ldots m$, we write $x^{(1:m)}\defeq (x^{(1)}, \ldots, x^{(m)})$.
Throughout, $\nu$ will denote the target probability of interest, and for $\vp\in L^1(\nu)$ we set $\bar{\vp}\defeq \vp-\nu(\vp)$, element of $L^1_0(\nu)$.    
  
  \subsection{Definitions}
Let $\mu$ and $\nu$ be $\gs$-finite measures on $\mathbf{X}$.  If $\mu(A)=0$ implies $\nu(A)=0$ for all $A\in \mathcal{B}(\mathbf{X})$, we say that $\nu$ is \emph{absolutely continuous} with respect to $\mu$, and write $\nu \ll \mu$.  Suppose $\nu\ll \mu$.  Recall that a \emph{Radon-Nikod\'{y}m derivative} of $\nu$ with respect to $\mu$ is a non-negative measurable function $\der{\nu}{\mu}(x)$ on $\mathbf{X}$ such that $\mu( \der{\nu}{\mu} g) = \nu(g)$ for all $g\in L^1(\nu)$.  If also $\mu$ and $\nu$ are probability densities, then it is easy to see that $\der{\nu}{\mu}(x)$ is in $L^1(\mu)$, and is equivalent with $\fr{\nu(x)}{\mu(x)}$.

Let $\mu$ be a probability on $\mathbf{X}$.  A Markov chain $K$ on $\mathbf{X}$ is $\mu$-\emph{invariant} if $\mu K=\mu$.  If also $\inner{f}{Kf}_\mu\ge 0$ for all $f\in L^2(\mu)$, then $K$ is \emph{positive}.  If $\mu( \ud x) K_x (\ud x') = \mu(\ud x') K_{x'}( \ud x)$, then $K$ is said to satisfy \emph{detailed balance} with respect to $\mu$, or briefly, $K$ is $\mu$-\emph{reversible}.  This implies that $K$ is $\mu$-invariant, and that the \emph{Dirichlet form} $\cE_K(f)$ for $f\in L^2(\mu)$ satisfies 
\begin{equation}\label{eq:dirichlet-reversible}
  \cE_K(f)
  \defeq
  \inner{f}{(1-K)f}_{\mu}
  = \fr{1}{2} \int \mu (\ud x) K_x( \ud x') \big( f(x) - f(x')\big)^2.
\end{equation}
We say Markov chain $K$ is $\mu$-\emph{Harris ergodic} if $K$ is $\mu$-invariant, $\psi$-irreducible, and Harris recurrent.
See \cite{meyn-tweedie} for the definition of $\psi$-irreducibility and Harris recurrence, and further details.  Most MCMC schemes are Harris ergodic, although a careless implementation can lead to a non-Harris chain \citep[see][]{roberts-rosenthal-harris}.

\section{Peskun type ordering for normalised importance sampling}\label{sec:normalised}
\subsection{General case}
Let $\mu$ and $\nu$ be probability measures on a measurable space $\mathbf{X}$, and let $\cv:\mathbf{X}\ra [0,\infty)$ be a non-negative measurable function.
  \begin{assumption}[Importance sampling]\label{a:is} A triplet $(\mu,\nu, \cv)$ is such that $\nu \ll \mu$ and $\cv(x)=\der{\nu}{\mu}(x)$ is the Radon-Nikod\'{y}m derivative. 
  \end{assumption}

  \begin{assumption}\label{a:peskun}
A heptuple $(\mu,\nu,\cv,K,L,\lc,\uc)$ is such that $(\mu,\nu,\cv)$ satisfies Assumption \ref{a:is}, $K$ and $L$ are Harris ergodic Markov chains reversible with respect to $\mu$ and $\nu$, respectively, and the constants $\lc,\,\uc\ge 0$ satisfy 
    \begin{enumerate}[(a)]
      \item\label{a:peskun-one}
        $\lc\,\cE_K(g) \le \cE_L(g) \le \uc\, \cE_K(g)$, for all $g\in L^2(\mu)$, and 
      \item\label{a:peskun-two}
        $\lc\le \cv \le \uc$, $\mu$-a.e. 
    \end{enumerate}
  \end{assumption}

\begin{theorem}\label{thm:peskun}
  If Assumption \ref{a:peskun} holds, then for all $\vp \in L^2(\nu)$,
\begin{align}\label{thm:peskun-equation}
\Var(K, \cv \vp) + \Var_{\mu}( \cv \bvp) &\le \uc \,\big[\Var(L,\vp) + \Var_{\nu}(\vp) \big], 
\\ \label{thm:peskun-equation-reverse}
  \Var(K, \cv \vp) + \Var_\mu(\cv \bvp)
  &\ge
  \lc\, \big[ \Var(L, \vp) + \Var_\nu( \vp)\big]. 
\end{align}
  \end{theorem}
\begin{remark} \label{thm:peskun-remark}
  Here, we recall the notation $\bar{\vp}\defeq \vp-\nu(\vp)$.
    Regarding Theorem \ref{thm:peskun}, whose proof is given in  Appendix \ref{app:peskun}:
    \begin{enumerate}[(i)]
    \item If $\cv=1$ constant, in which case $\mu=\nu$, it reduces to \citep[Lemma 32]{andrieu-lee-vihola}.  If also $(\lc,\uc)=(0,1)$, it is the covariance ordering \citep[Thm.~4.2]{mira-geyer-ordering}, which is a  Peskun  \cite{peskun,tierney-note} type criterion based on the  Dirichlet form \citep[see also][Proof of Lem.~3]{tierney-note}.
    \item The assumptions are the same as those of \citep[Lem. 13.22]{levin-peres-wilmer} about comparison of mixing times in the countable state space context.
    \item \eqref{thm:peskun-equation-reverse} holds even if we `forget' $\uc$, i.e.~set $\uc=\infty$ but also require $\cv \vp\in L^2(\mu)$. 
      Unless $\mathbf{X}$ is compact, \eqref{thm:peskun-equation-reverse} is usually redundant since we can only assume $\lc=0$.
\item At least in the examples we will consider, one can take $\uc\defeq \norm{\cv}_{L^\infty(\mu)}$ and $\lc\defeq \mu-\essinf \cv$ to satisfy Assumption \ref{a:peskun}(\ref{a:peskun-one}-\ref{a:peskun-two}) (see Remark \ref{lem:p-remark}\eqref{lem:p-remark-constant}).
    \end{enumerate}
  \end{remark}

  \subsection{Marginalisations and augmented importance sampling kernels}
  Let $\mathbf{X}=\mathbf{T}\times\mathbf{Y}$ be a joint space.  For a probability $\mu$ on $\mathbf{X}$, denote by $\mu^*(\ud \gt) = \mu(\ud \gt, \mathbf{Y})$ its marginal probability.  If $(\mu,\nu,\cv)$ on $\mathbf{X}$ satisfies Assumption \ref{a:is}, then $\nu^*\ll \mu^*$, and with $\cv^*(\gt)\defeq \der{\nu^*}{\mu^*}(\gt)$, the triplet $(\mu^*, \nu^*, \cv^*)$ satisfies Assumption \ref{a:is} on $\mathbf{T}$.  
  \begin{assumption}\label{a:peskun-marginal}
Assumption \ref{a:peskun}, with Assumption \ref{a:peskun}(\ref{a:peskun-one}--\ref{a:peskun-two}) replaced with
    \begin{enumerate}[(a)]
      \item\label{a:peskun-marginal-one}
        $\lc \,\cE_K(g) \le \cE_L(g) \le \uc \, \cE_K(g)$, for all $g\in L^2(\mu^*)$, and 
      \item\label{a:peskun-marginal-two}
        $\lc\,\le \cv^* \le \uc$, $\mu^*$-a.e.
    \end{enumerate}
  \end{assumption}

  We introduce the notion of an augmented Markov kernel, as in \cite{andrieu-vihola-order,vihola-helske-franks}.
  \begin{definition}\label{def:augmented}
    Let $\dot{\mu}$ be some probability on $\mathbf{T}$, let $\dot{K}$ be a $\dot{\mu}$-invariant Markov kernel on $\mathbf{T}$, and let $Q_{\gt}(\ud y)$ be a probability kernel from $\mathbf{T}$ to $\mathbf{Y}$.  The \emph{$Q$-augmentation of $\dot{K}$}, or the \emph{$Q$-augmented kernel} $K$, is a Markov kernel on $\mathbf{X}$, with transition $K$ and invariant measure $\mu$, given by
  \begin{equation}\label{eq:augmented}
    K_{\gt\, y} (\ud \gt', \ud y') = \dot{K}_{\gt}(\ud \gt') Q_{\gt'}(\ud y'),
    \quad\text{and}\quad
    \mu(\ud \gt,\ud y) = \dot{\mu}(\ud \gt) Q_{\gt}( \ud y).
  \end{equation}
  
  \end{definition}

  \begin{theorem}\label{thm:peskun-augmented}
    Suppose Assumption \ref{a:peskun-marginal} holds, and that $K$ is an augmented kernel as in Definition \ref{def:augmented}. Let $\vp\in L^2(\nu)$ with $\cv \vp\in L^2(\mu)$.  With $\mathscr{N}_K\defeq 0$ if $K$ is positive, and $\mathscr{N}_K\defeq 1$ if not, the following bound holds:
\begin{equation}\label{eq:peskun-augmented}
  \Var(K, \cv \vp) \le \uc \big[ \Var(L,\vp)+ \Var_\nu(\vp)\big]  + (1+2\mathscr{N}_K)\, \Var_\mu(\cv \bar{\vp})
\end{equation}
Moreover, if $\cv \vp$ only depends on $\gt\in \mathbf{T}$, then \eqref{thm:peskun-equation} holds.
\end{theorem}
  \begin{remark}
    \label{thm:peskun-augmented-remark}
    Regarding Theorem \ref{thm:peskun-augmented}, whose proof is given in Appendix \ref{app:peskun}:
\begin{enumerate}[(i)]
\item \label{thm:peskun-augmented-remark-auxiliary} The function $\vp$ (and $\cv \vp$) is allowed to depend on the auxiliary variable $y\in \mathbf{Y}$, unlike comparison results in the PM setting (see~\citep[Thm.~7]{andrieu-vihola-convergence} and \citep[Thm.~1]{sherlock-thiery-lee}) that are based on the convex order \citep[Thm.~10]{andrieu-vihola-order}.
  
\item\label{thm:peskun-augmented-remark-positive} $K$ is positive iff $\dot{K}$ is positive (Lemma \ref{lem:augmented-kernel} of Appendix \ref{app:misc}).  This is the case e.g.~if $\dot{K}$ is a random walk $\MH$ kernel with normal proposals \citep[Lem. 3.1]{baxendale}.  See \cite[Prop.~3]{doucet-pitt-deligiannidis-kohn} for more examples.
\item See also Remarks \ref{lem:p-remark}(\ref{lem:p-remark-constant}--\ref{lem:p-remark-subspace}) in Appendix \ref{app:peskun} about Assumption \ref{a:peskun-marginal}, which also hold for Assumption \ref{a:peskun} by trivialising the space $\mathbf{Y}$ (Lemma \ref{lem:helper}\eqref{lem:helper-one}).  
  \end{enumerate}
\end{remark}

\section{Jump chains and self-normalised importance sampling}\label{sec:jump}        \subsection{Jump chains}\label{sec:jump-jump}    
    We recall the notion of a jump chain \citep[see][]{douc-robert}, which is a Markov chain consisting of the accepted states of the original chain.
\begin{definition}\label{def:jump-chain}
  Let $(\gT_k)_{k\ge 1}$ be a Markov chain with transition $K_{\gt}(\ud \gt')$.  The \emph{jump chain} $(\tilde{\gT}_k, \tilde{N}_k)_{k\ge 1}$ with transition $\tilde{K}_{\gt n}(\ud \gt', \ud n')$ and holding times
  $$
  \tilde{N}_j \defeq \min\Big\{i\ge 1| \gT_{\tilde{N}^*_{j-1} +i +1}  \neq \gT_{\tilde{N}^*_{j-1} +1} \Big\}
  ,\qquad j\ge 1,
  $$
  is given by $\tilde{\gT}_1\defeq \gT_1$ and $\tilde{\gT}_{k+1} \defeq \gT_{\tilde{N}^*_k+1}$, where $\tilde{N}^*_k \defeq\sum_{j=1}^k \tilde{N}_j$, $\tilde{N}^*_0 \defeq 0$.
\end{definition}
For a Harris ergodic chain $K$, $(\tilde{N}_k)_{k\ge 1}$ are independent random variables given $(\tilde{\gT}_k)_{k\ge 1}$, where $\tilde{N}_k$ is geometrically distributed with parameter $\ga(\tilde{\gT}_k)$.  Here, $\ga(\gt)\defeq K(\gt, \mathbf{T}\backslash\{\gt\})$ is the acceptance probability function of $K$ at $\gt\in \mathbf{T}$.  See \citep[Prop.~24]{vihola-helske-franks} for this as well as for proof of the following result.
\begin{lemma}\label{lem:jump-chain-kernel}
Let $K$ be a $\mu$-invariant Markov chain with $\ga>0$.  The marginal chain $\tilde{K}$ of the jump chain of $K$ has transition $\tilde{K}(\gt,A) = K(\gt,A\backslash\{\gt\})/\ga(\gt)$, for all $A\in \mathcal{B}(\mathbf{T})$, and is $\tilde{\mu}$-invariant, where $\tilde{\mu}(\ud \gt)= \ga(\gt) \mu(\ud \gt)/\mu(\ga)$.  Moreover, $K$ is $\mu$-reversible iff $\tilde{K}$ is $\tilde{\mu}$-reversible, and $K$ is $\mu$-Harris ergodic iff $\tilde{K}$ is $\tilde{\mu}$-Harris ergodic.
\end{lemma}
We note that $(\tilde{\gT}_k,\tilde{N}_k)_{k\ge 1}$ has as its transition the $Q^{(N)}$-augmentation of $\tilde{K}$ (Definition \ref{def:augmented}), where $\tilde{K}$ is as in Lemma \ref{lem:jump-chain-kernel} and
$Q_{\gt}^{(N)}(\cdot) \sim \mathrm{Geo}(\ga(\gt))$ \citep[][]{doucet-pitt-deligiannidis-kohn}.

Different estimators can sometimes be used in place of $(\tilde{N}_k)$, which can lead to lower asymptotic variance of the related MCMC than when not using the jump chain, or when using the jump chain with standard $(\tilde{N}_k)$ \cite{douc-robert}. 
\subsection{Self-normalised importance sampling}\label{sec:jump-snis}
Jump chains can be naturally used with IS estimators, and can lead to improved computational and statistical efficiency \citep[see][]{vihola-helske-franks}.  To avoid redundancy, we shall adhere to the following convention: when we write $(\mathsf{\gT}_k, \mathsf{N}_k, \mathsf{a}, \upmu)$, it shall stand simultaneously for
$(\tilde{\gT}_k, \tilde{N}_k, \ga,\tilde{\mu})$, corresponding to an IS jump chain (denoted `ISJ'), and for $(\gT_k, 1,1,\mu)$, corresponding to a non-jump IS chain (denoted `IS0').  

Suppose $(\mu, \nu, \cv)$ satisfies Assumption \ref{a:is} and that $(\gT_k)_{k\ge 1}$ is $\mu$-Harris ergodic.  Often one can not evaluate $\cv(\gt)$.  However, one can often evaluate an unnormalised version $\cvu(\gt)= c_\xi\cdot \cv(\gt)$, with $c_{\xi}>0$ a (unknown) constant.  In this case, for $\vp \in L^1(\nu)$, one can use the following SNIS estimator, 
\begin{equation}\label{eq:snis-estimator}
E_n^{SNIS}(\vp) \defeq \fr{ \sum_{k=1}^n \mathsf{N}_k \cvu(\mathsf{\gT}_k) \vp(\mathsf{\gT}_k)}{\sum_{k=1}^n \mathsf{N}_k\cvu(\mathsf{\gT}_k)}
=
 \fr{ \fr{1}{n}\sum_{k=1}^n \mathsf{N}_k \cvu(\mathsf{\gT}_k) \vp(\mathsf{\gT}_k)}{\fr{1}{n} \sum_{k=1}^n \mathsf{N}_k\cvu(\gT_k)}.
\end{equation}
By Harris ergodicity, the SNIS estimator is a consistent estimator for $\nu(\vp)$,
$$
E_n^{SNIS}(\vp)
\xrightarrow[\text{a.s.}]{n\to\infty}
 \fr{ \upmu( \E[\mathsf{N}_k |\mathsf{\gT}_k] \cvu \vp)}{\upmu( \E[\mathsf{N}_k| \mathsf{\gT}_k] \cvu)}
=
\fr{ \upmu( \cvu \vp/ \mathsf{a})}{\upmu( \cvu/ \mathsf{a})}
=
\nu(\vp).
$$

Next we consider a framework on an extended space, from which a Peskun type ordering for SNIS will trivially follow (Remark \ref{thm:peskun-is-remark}\eqref{thm:peskun-is-remark-snis} of Theorem \ref{thm:peskun-is}).  

\section{Unbiased estimators and exact approximation schemes}\label{sec:schemes}
In many settings, one relies on unbiased estimators in the MCMC \citep[see][]{andrieu-roberts}.  We now describe a framework of unbiased estimators which we use and which is suitably general \citep[see][]{vihola-helske-franks,andrieu-doucet-holenstein}.  We then describe direct and IS MCMC schemes to calculate $\nu(f)$, and give a general comparison result of their asymptotic variances.

\subsection{Framework}
Recall from Section \ref{sec:intro} that our goal is calculation of expectations
$
\nu(f) = \int f(\gt, z) \nu(\ud \gt, \ud z),
$
with respect to the joint probability $\nu$ on $\mathbf{T}\times \mathbf{Z}$, as well as with respect to its marginal probability $\dot{\nu}(\ud \theta) = \nu(\ud \theta, \mathbf{Z})$.  For our unbiased estimators, define the spaces 
\begin{align*}
  \mathbf{U} &\defeq \big\{ (\ell, \eta^{(1:\ell)}): \ell \in \mathbb{N}, \eta^{(i)}\in [0,\infty), \text{ for } i=1,\ldots, \ell\big\}\\
\mathbf{V}
&\defeq
\big\{ (m,z^{(1:m)}, \zeta^{(1:m)}) : m\in \mathbb{N},\text{ and } z^{(i)}\in \mathbf{Z}, \zeta^{(i)}\in [0,\infty)\text{ for } i=1,\dots m
  \big\}.
\end{align*}
Let $Q_\theta^{(U)}(\ud u)$ be a probability on $\mathbf{U}$ for each $\theta\in\mathbf{T}$, and $Q_{\theta u }^{(V)}(\ud v)$ a probability on $\mathbf{V}$ for each $(\theta,u)\in \mathbf{T}\times\mathbf{U}$.  
Given $\Theta_k\in\mathbf{T}$, $(\ell_k, \eta_k^{(\ell_k)})\in\mathbf{U}$, $(M_k, Z_k^{(1:m)}, \zeta_k^{(1:m)})\in \mathbf{V}$, and a function $f$ on $\mathbf{T}\times\mathbf{V}$, we define formally
\begin{equation}\label{filtering}
\zeta_k(f)\defeq \frac{1}{M_k} \sum_{i=1}^{M_k} \zeta_k^{(i)} f(\Theta_k, Z_k^{(i)}),
\qquad
\hat{\zeta}_k(f) = \frac{\zeta_k(f)}{\zeta_k(1)},
\qquad
\xi_k \defeq \frac{\zeta_k(f)}{\eta_k(1)},
\end{equation}
where $\eta_k(1)\defeq \frac{1}{\ell_k}\sum_{i=1}^{\ell_k} \eta_k^{(i)}$. 
Let $\mathcal{L}^1(\nu)$ denote the set of functions $f$ on $\mathbf{T}\times \mathbf{Z}$ such that there exists a constant $c_\zeta>0$ such that for all $\theta\in\mathbf{T}$ and $f\in \mathcal{L}^1(\nu)$,
    $$
\int Q_\theta^{(U)}(\ud u) Q_{\theta u }^{(V)}(\ud v) \zeta(g) = c_\zeta \dot{\nu}(\theta) \int g(\theta, z) \nu(\ud z| \theta) 
$$
for $g\in \{f, |f|\}$, where $\nu(\ud z| \theta)$ denotes a regular conditional probability of $\nu$ given $\theta$.  Also, define $\mathcal{L}^2(\nu) = \{f\in \mathcal{L}^1(\nu): f^2\in \mathcal{L}^1(\nu)\}$.  

\begin{assumption}\label{a:pm-kernels}
The following hold:
  \begin{enumerate}[(i)]
  \item The constant function $1\in \mathcal{L}^1(\nu)$.
\item \label{a:pm-kernels-support} For all $\Theta_k\in \mathbf{T}$, the variables $U_k=(\ell_k, \eta_k^{(1:\ell_k)})\sim Q_{\Theta_k}^{(U)}(\cdot)$ and $V_k= (M_k, Z_k^{(1:M_k)}, \zeta_k^{(1:M_k)})\sim Q_{\Theta_k U_k}^{(V)}(\cdot)$ satisfy
  $$
\eta_k(1) =0 \implies \zeta_k(1) =0.
  $$
    \end{enumerate}

  \end{assumption}

\begin{remark} Regarding Assumption \ref{a:pm-kernels} and the above definitions:
  \label{a:pm-kernels-remark}
\begin{enumerate}[(i)]
\item  If $f \in L^1(\nu)$ satisfies $f(\gt,\cdot) = f(\gt)$, then $f\in \mathcal{L}^1(\nu)$.  In many settings, $\mathcal{L}^1(\nu)$ may be much larger, or all of $L^1(\nu)$ \citep[see][Cor.~28]{vihola-helske-franks}.
\item\label{a:pm-kernels-remark-support} Support condition \eqref{a:pm-kernels-support} holds quite generally, e.g.~if $\eta(1)>0$.  In a setting where, given $\gt$, $\eta(1)=\pr(\theta) \eta'(1)$ and $\eta'(1)$ is an unbiased estimator for an approximate likelihood $L^{(U)}(\gt)$, this can be achieved by inflating the likelihood $L^{(U)}(\theta)$ and $\eta'(1)$ uniformly by a constant $\epsilon>0$: $L^{(U)}(\gt)\mapsto L^{(U)}(\gt) + \epsilon$ and $\eta'(1)\mapsto \eta'(1) + \epsilon$ for all $\theta$ \citep[see][Prop.~19 and Rem.~20]{vihola-helske-franks}. 
\end{enumerate}
\end{remark}

\subsection{Pseudo-marginal type schemes and importance sampling schemes}
Define the probability
$
\pi(\ud \theta, \ud u, \ud v)
\defeq
c_\zeta^{-1} \ud \theta Q_\theta^{(U)}(\ud u) Q_{\theta u}^{(V)} (\ud v) \zeta(1).
$
The following concerns a PM/DA type scheme \citep[][]{andrieu-roberts,christen-fox}.
\begin{proposition}\label{prop:pm-scheme}
Suppose a Markov chain $(\gT_k, U_k, V_k)_{k\ge 1}$ is $\pi$-reversible Harris ergodic, where Assumption \ref{a:pm-kernels} holds. 
Then, for all $f\in \mathcal{L}^1(\nu)$,
  \begin{equation}\label{pm-estimator}
  E_n^{PM}(f) \defeq \fr{1}{n} \sum_{k=1}^n \hat{\zeta}_k(f)
\xrightarrow[\text{a.s.}]{n\to\infty}
  \nu(f).
  \end{equation}
      \end{proposition}
\begin{proof}
Follows by Harris ergodicity, as $\pi(\hat{\zeta}(f))=\nu(f)$, $f\in \mathcal{L}^1(\nu)$.  
  \end{proof}

Define the probability
$
\mu(\ud \theta, \ud u) \defeq
c_\eta^{-1} \ud \theta Q_{\theta}^{(U)}(\ud u) \eta(1)
$
where $c_\eta>0$ is a normalising constant.  Set $\dot{\mu}(\ud \theta) = \mu(\ud \theta, \mathbf{U})$.  Note that under Assumption \ref{a:pm-kernels}, we have $\dot{\mu}(\theta)=0$ implies $\dot{\nu}(\theta)=0$.  

Consider now an IS scheme (Algorithm \ref{alg:is-scheme}) as in \citep[][]{vihola-helske-franks}.  Compared to \cite{vihola-helske-franks}, we additionally assume $\mu$-reversibility of the base chain and nonnegativity of the estimators $\zeta^{(i)}\ge 0$.  This is done to facilitate comparison with the previous PM type scheme corresponding to PM and DA algorithms, which are $\pi$-reversible and require $\zeta^{(i)}\ge 0$, as $\zeta(1)$ is present in their acceptance ratio (see Section \ref{sec:application}).
\begin{algorithm}
  \caption{Importance sampling scheme.  Suppose Assumption \ref{a:pm-kernels} holds.} \label{alg:is-scheme}
  \begin{enumerate}[(Ph{a}se 1)]
    \item
      Let $(\gT_k, U_k)_{k\ge 1}$ be a $\mu$-reversible Harris ergodic Markov chain.
\item 
      For each $k\ge 1$, let $V_k$ be drawn as follows, for the IS0 and ISJ cases:
      \begin{enumerate}
      \item[(IS0)] $V_k\sim Q^{(V)}_{\gT_k U_k}(\cdot)$.  For $f\in \mathcal{L}^1(\nu)$, we define
        \begin{equation}\label{alg:is-scheme-iso-equation}
        \mathbf{m}_f(\gt, u)
        \defeq
        \E[\xi_k(f) | \gT_k=\gt, U_k=u].
        \end{equation}
      \item[(ISJ)]  Form a jump chain $(\tilde{\gT}_k,\tilde{U}_k,\tilde{N}_k)_{k\ge 1}$, and draw $V_k$ from some kernel $V_k\sim Q^{(V|N)}_{\tilde{\gT}_k \tilde{U}_k \tilde{N}_k}(\cdot)$ from $\mathbf{T}\times \mathbf{U}\times\mathbb{N}$ to $\mathbf{V}$ such that
        $$
\E[ \xi_k(f)| \tilde{\gT}_k=\gt, \tilde{U}_k=u, \tilde{N}_k=n] = \mathbf{m}_f(\gt,u)
        $$
for all $n\in\mathbb{N}$ and $f\in \mathcal{L}^1(\nu)$.         
        \end{enumerate}      
  \end{enumerate}
  \end{algorithm}
    If Assumption \ref{a:pm-kernels} (PM kernels) holds, then for all $f\in \mathcal{L}^1(\nu)$,
    $$
    \mu( \mathbf{m}_f )
    =
    \frac{1}{c_\eta} \int \ud \gt Q^{(U)}_\gt(\ud u) \eta(1) Q^{(V)}_{\gt u}(\ud v) \xi(f)
= c_\xi \nu(f)    
    $$
where $c_\xi\defeq c_\zeta/c_\eta$, and $\mathbf{m}_f$ is defined in \eqref{alg:is-scheme-iso-equation}.  This motivates the following consistency result, an instance of \citep[Thm.~3]{vihola-helske-franks} for example for the $\mathsf{N}_k=1$ case (IS0) and \citep[Thm.~13]{vihola-helske-franks} for the $\mathsf{N}_k=\tilde{N}_k$ case (ISJ).  

\begin{proposition}\label{prop:is-consistency}
  Under Algorithm \ref{alg:is-scheme}, for all $f\in \mathcal{L}^1(\nu)$,
  \begin{equation}\label{is-estimator}
  E_n^{\mathrm{IS}}(f)
  \defeq
  \fr{\sum_{k=1}^n \mathsf{N}_k \xi_k(f)}{\sum_{k=1}^n \mathsf{N}_k \xi_k(1)}
  \xrightarrow[\text{a.s.}]{n\to\infty}
   \nu(f).
  \end{equation}
\end{proposition}
\begin{remark}
    In the ISJ case, permitting dependence on $\tilde{N}_k$ when drawing $V_k$ in Algorithm \ref{alg:is-scheme} allows for variance reduction of $\xi_k(f)$ and hence of the resultant estimator \eqref{is-estimator} (see Proposition \ref{prop:is-asvar}), by using larger $M_k$ when $\tilde{N}_k$ is large.  For example, $M_k$ could correspond to the number of independent samples drawn from an instrumental or to the number of particles used in a particle filter \cite{andrieu-doucet-holenstein}.
  \end{remark}

\subsection{A Peskun type ordering for importance sampling schemes}
Under Assumption \ref{a:is-clt} below, the IS estimator $E^{\IS}_n(f)$ \eqref{is-estimator} satisfies a CLT
\begin{equation}\label{is-clt}
  \sqrt{n}[E^{\IS}_n(f) - \nu(f)]
  \xrightarrow{n\ra\infty}
  \mathcal{N}\big(0, \mathbb{V}^{\IS}_f\big),
\qquad\text{in distribution}.
\end{equation}
See \cite{vihola-helske-franks} or Proposition \ref{prop:is-asvar} of Appendix \ref{app:peskun}, with a formula for $\mathbb{V}^{\IS}_f$.  In analogy with Definition \ref{def:asvar} and \eqref{eq:kipnis-varadhan}, we refer to $\mathbb{V}^{\IS}_f$ as the \emph{IS asymptotic variance}.  The following assumption is sufficient for $\mathbb{V}_f^{\IS} <\infty$.    
\begin{assumption}[Importance sampling CLT]\label{a:is-clt}
  Suppose Algorithm \ref{alg:is-scheme} (IS scheme) and that $(\mathsf{\gT}_k, \mathsf{U}_k, \mathsf{N}_k)_{k\ge 1}$ is aperiodic.  
  Let $f\in \mathcal{L}^2(\nu)$ be a function such that
  $\Var(K, \mathbf{m}_f)<\infty$,
  where 
  $\mathbf{m}_f$ is defined in \eqref{alg:is-scheme-iso-equation},
  and $\mathsf{v}_{\bar{f}}$ by
  \begin{enumerate}
  \item[(IS0)] $v_{\bar{f}}(\gt,u) \defeq \Var\big(\xi_k(\bar{f})|\gT_k= \gt, U_k=u\big)$,
  \item[(ISJ)]
$\tilde{v}_{\bar{f}} (\gt,u) \defeq \E[ \tilde{N}_k^2 \Var\big(\xi_k(\bar{f} ) |\tilde{\gT}_k =\gt, \tilde{U}_k=u, \tilde{N}_k \big)| \tilde{\gT}_k =\gt, \tilde{U}_k=u],
    $
    \end{enumerate}
satisfies 
  $
  \mu\big( \mathsf{a}\mathsf{v}_{\bar{f}} \big) <\infty.
  $
\end{assumption}
Let us denote the kernel and measure of the IS0 corrected chain of Algorithm \ref{alg:is-scheme} by $(\bar{K}, \bar{\mu})$ on the space $\mathbf{X}=(\mathbf{T}\times\mathbf{U})\times\mathbf{V}$, where,
\begin{align}\label{is-kernel}
\bar{K}_{\gt u v}(\ud \gt', \ud u', \ud v')
&\defeq
K_{\gt u}(\ud \gt', \ud u') Q^{(V)}_{\gt' u'}(\ud v')
\nonumber\\
\bar{\mu}(\ud \gt, \ud u, \ud v)
&\defeq
\mu(\ud \gt, \ud u) Q^{(V)}_{\gt u}(\ud v).
\end{align}
Note that $\bar{K}=K^{(V)}$ is an augmented kernel (Definition \ref{def:augmented}).  Note too that by Slutsky's lemma like in \eqref{eq:slutsky}, we have $\mathbb{V}_f^{\IS}=\Var\big(\bar{K}, w\hat{\zeta}(f)\big)$, where $w=\ud \pi/ \ud \bar{\mu}$.  

With definitions as in Assumption \ref{a:is-clt}, we define a `difference' constant $\mathsf{D}_{\bar{f}}$, for the IS0 and ISJ cases, respectively, by $D_{\bar{f}}\defeq 0$ and 
$$\tilde{D}_{\bar{f}}
\defeq
\mu(a) c_\xi^{-2} \mu( a \tilde{v}_{\bar{f}} - v_{\bar{f}}).$$ 

\begin{theorem}\label{thm:peskun-is}
Suppose the assumptions of Algorithm \ref{alg:is-scheme} (IS scheme) hold, and $\mathbb{V}^{\IS}_f<\infty$.  
\begin{enumerate}[(i)]
\item \label{thm:peskun-is-one}
  If $(\bar{\mu}, \pi, \cv, \bar{K}, L, \lc,\uc)$ satisfies Assumption \ref{a:peskun} on $\mathbf{X}$, then
    \begin{align*}
      \mathbb{V}^{\IS}_{f}
          + \mu(\mathsf{a}) \Var_{\bar{\mu}}\big( \cv \hat{\zeta}(\bar{f}) \big)
    &\le \uc\,
    \mu(\mathsf{a})  \big\{\Var\big(L, \hat{\zeta}(f) \big) + \Var_\pi\big(\hat{\zeta}(f)\big)\big\}
    + \mathsf{D}_{\bar{f}}
    \\
    \mathbb{V}^{\IS}_{f}
    + \mu(\mathsf{a}) \Var_{\bar{\mu}}\big( \cv \hat{\zeta}(\bar{f}) \big)
    &\ge \lc\,
\mu(\mathsf{a})  \big\{\Var\big(L, \hat{\zeta}(f) \big) + \Var_\pi\big(\hat{\zeta}(f)\big)\big\}
    + \mathsf{D}_{\bar{f}}.
    \end{align*}
  \item \label{thm:peskun-is-two}
    If $(\bar{\mu}, \pi, \cv, \bar{K}, L, \lc, \uc)$ satisfies Assumption \ref{a:peskun-marginal} on $\mathbf{X}$, then
    \begin{align}
\nonumber
    \mathbb{V}^{\IS}_f \le
    &\,\uc\, \mu(\mathsf{a}) \big\{  \Var\big(L,\hat{\zeta}(f) \big) + \Var_\pi\big(\hat{\zeta}(f)\big)\big\}
    \\ \nonumber
    &+ (1+2\mathscr{N}_K)\mu(\mathsf{a}) \Var_{\bar{\mu}}\big( \cv \hat{\zeta}(\bar{f})\big)
    + \mathsf{D}_{\bar{f}}
    \end{align}
  where $\mathscr{N}_K\defeq 0$ if $K$ is positive, and $\mathscr{N}_K \defeq 1$ if not.  
\end{enumerate}
\end{theorem}

\begin{remark}\label{thm:peskun-is-remark}
  Regarding Theorem \ref{thm:peskun-is}, whose proof is in Appendix \ref{app:peskun}:
  \begin{enumerate}[(i)]
  \item \label{thm:peskun-is-remark-definitions}
    Note that $0\le \mu(\mathsf{a}) \le 1$, with $\mathsf{a}$ as in Section \ref{sec:jump-snis}, and that $\cv= c_\xi\inv \xi(1)$ and $\cv^* = c_\xi\inv \mathbf{m}_1$, with $\mathbf{m}_f(\gt,u)$ defined in \eqref{alg:is-scheme-iso-equation}.
\item \label{thm:peskun-is-remark-snis}  As a trivialisation, when $\eta(\gT_k, U_k)\defeq \eta(1) = \dot{\mu}(\gT_k)$ a.s., $\mathbf{Z}=\{0\}$, and $\xi_k(f)= \cv_u(\gT_k)f(\gT_k)$ a.s., we obtain a Peskun type ordering for SNIS \eqref{eq:snis-estimator}.  Here, the simplifications are $\bar{K}\leftrightarrow K$, $\hat{\zeta}(\bar{f})\leftrightarrow \bar{f}$ and $\xi(\bar{f}) \leftrightarrow  c_{\xi} \cv \bar{f}$.  
  \end{enumerate}  
  \end{remark}

\section{Pseudo-Marginal and delayed-acceptance MCMC}\label{sec:application}
We define PM and DA type algorithms in the setting of the auxiliary variable framework of Section \ref{sec:schemes}, where PM could be the `particle marginal $\MH$' \cite{andrieu-doucet-holenstein}; a DA type variant of this algorithm has been implemented e.g.~in \cite{golightly-henderson-sherlock,quiroz-tran-villani-kohn,vihola-helske-franks}.  After defining the corresponding kernels, we then compare the asymptotic variances of PM/DA with IS (Theorem \ref{thm:comparison}).  
\subsection{Algorithms}\label{sec:application-algorithms}
Let $q_{\gt}(\ud \gt') = q_{\gt}(\gt') \ud \gt'$ be a proposal kernel on $\mathbf{T}$.  Assume the setup of Assumption \ref{a:pm-kernels} (recall that $\eta(1)\ge 0$ and $\zeta(1) \ge 0$). Whenever the denominators are not zero we define the following `acceptance ratios' for $x,x'\in \mathbf{X} \defeq \mathbf{T}\times \mathbf{U}\times\mathbf{V}$, where $x=(\gt, u, v)$,
\begin{equation}\label{acceptance-ratios}
\ar{U}{x}{x'}\defeq \fr{ \eta'(1) q_{\gt'}( \gt)}{\eta(1) q_{\gt}( \gt')},
  \quad\text{and}\quad
  \ar{V}{x}{x'}\defeq \fr{\zeta'(1) q_{\gt'}( \gt)}{\zeta(1) q_{\gt}( \gt')}.
\end{equation}

Consider Algorithm \ref{alg:pm-parent} (`PM parent,' following the terminology of \cite{sherlock-lee}), Algorithm \ref{alg:dapr} (`$\DApr$'), and Algorithm \ref{alg:da} (`$\DAg$'), with transition kernels given later and which are $\pi$-invariant \citep[see][]{andrieu-doucet-holenstein,andrieu-roberts,banterle-grazian-lee-robert}.    Under Assumption \ref{a:pm-kernels} (PM kernels) and the assumption that the resultant chains are $\pi$-Harris ergodic, by construction Algorithms (\ref{alg:pm-parent}-\ref{alg:da}) produce output as in Proposition \ref{prop:pm-scheme} (PM type scheme).  In PM parent (Algorithm \ref{alg:pm-parent}) and $\DAg$ (Algorithm \ref{alg:da}), the computationally expensive $V_k$-variable is drawn whenever $U_k$ is drawn.
This is the essential difference with $\DApr$ (Algorithm \ref{alg:dapr}).  The separation of sampling steps can substantially reduce computational cost in $\DApr$ \citep[see][]{christen-fox}, even though the asymptotic variance of $\DApr$ is more than PM parent in the case
$K$ is the approximate PM kernel \eqref{is-pm-kernel}
\citep[see][]{banterle-grazian-lee-robert}, and more than $\DAg$ in the case $K$ is a `$\mu$-proposal-rejection chain' (e.g.~PM); see Propositions \ref{prop:das} and \ref{prop:dag} below).

  \begin{algorithm}
    \caption{Pseudo-Marginal parent.  Suppose Assumption \ref{a:pm-kernels} (PM kernels) holds.  Initialise $X_0\in \mathbf{X}$ with $\zeta_0(1)>0$.  For $k=1,\ldots n$, do:}
\label{alg:pm-parent}
    \begin{enumerate}[(1)]
    \item\label{alg:pm-parent-one} Draw $\gT_{k}'\sim q_{\gT_{k-1}}(\cdot)$ and
$U_k' \sim Q_{\gT_k'}^{(U)}(\cdot)$ and $V_k'\sim Q_{\gT_k' U_k'}^{(V)}(\cdot)$.   With probability
      $
\min\big\{1, \ar{V}{X_{k-1}}{X_k'} \big\}
$
accept $X_k'$; otherwise, reject.  
  \end{enumerate}
  \end{algorithm}
\begin{algorithm}
  \caption{Delayed-acceptance ($\DApr$). Suppose Assumption \ref{a:pm-kernels} (PM kernels) holds, and $K$ is a $\mu$-proposal-rejection kernel of the form \eqref{pr-kernel}.  Initialise $X_0\in \mathbf{X}$ with $\zeta_0(1)>0$.  For $k=1,\ldots ,n$, do:} \label{alg:dapr}
  \begin{enumerate}[(1)]
  \item \label{alg:dapr-one}
    Draw $\gT_k' \sim q_{\gT_{k-1}}(\cdot).$
    Construct $U_k' \sim Q^{(U)}_{\gT_k'}(\cdot)$.  With probability $\alpha(\gT_{k-1}, U_{k-1}; \gT_k', U_k')$, proceed to step \eqref{alg:dapr-two}.  Otherwise, reject.
\item \label{alg:dapr-two}
  Construct $V_{k}'\sim Q^{(V)}_{\gT_k', U_k'}(\cdot)$. With probability $\min\big\{1, \xi_k'(1) /\xi_k(1)\big\}$, accept $(\gT_k', U_k', V_k')$; otherwise, reject. 
  \end{enumerate}
  \end{algorithm}
\begin{algorithm}
  \caption{Delayed-acceptance ($\DAg$). Suppose Assumption \ref{a:pm-kernels} (PM kernels) holds.  Initialise $X_0\in \mathbf{X}$ with $\zeta_0(1)>0$.  For $k=1,\ldots ,n$, do:} \label{alg:da}
  \begin{enumerate}[(1)]
  \item \label{alg:da-one}
    Draw $(\gT_k',U_k')\sim K_{\gT_{k-1},U_{k-1}}(\cdot)$.
\item \label{alg:da-two}
  Construct $V_{k}'\sim Q^{(V)}_{\gT_k', U_k'}(\cdot)$. With probability $\min\big\{1, \xi_k'(1) /\xi_k(1)\big\}$, accept $(\gT_k', U_k', V_k')$; otherwise, reject. 
  \end{enumerate}
  \end{algorithm}

\subsection{Kernels}\label{sec:application-kernels}
Let $K$ be the transition kernel of a $\mu$-reversible Harris ergodic IS0 base chain $(\gT_k,U_k)_{k\ge 1}$, with definitions as in Assumption \ref{a:pm-kernels} (PM kernels).  The \emph{$\DAg$ correction} of $K$ is the $\pi$-reversible kernel $K^{\DAg}$ corresponding to Algorithm \ref{alg:da}, given by, 
\begin{align}\label{dag-kernel}
K^{\DAg}_{\gt u  v}(\ud \gt', \ud u', \ud v')
&=
K_{\gt u}(\ud \gt', \ud u') Q^{(V)}_{\gt' u'}(\ud v')
\min\big\{ 1, \xi'(1)/ \xi(1) \big\}
\nonumber\\&\qquad\qquad
+[1-\ga_{\DAg}(\gt, u ,v)]\gd_{\gt u v}(\ud \gt',\ud u',\ud v'), 
\end{align}
where
$
\ga_{\DAg}(\gt,u,v) \defeq \int K_{\gt u}(\ud \gt', \ud u') Q^{(V)}_{\gt' u'}(\ud v')
\min\big\{ 1, \xi'(1)/ \xi(1) \big\}.
$

Let $K$ be `$\mu$-proposal-rejection kernel,' that is, a $\mu$-reversible kernel of the form 
\begin{equation}\label{pr-kernel}
K_{\theta u}(\ud \theta',\ud u')
=
q_\theta(\ud \theta') Q^{(U)}_{\theta'}(\ud u') \alpha(\theta, u; \theta', u')
+
r_K(\theta,u) \gd_{\theta u}(\ud \theta' \ud u')
\end{equation}
for some function $\alpha: (\mathbf{T}\times\mathbf{U})^2 \rightarrow [0,1]$ and 
$
r_K=
1- \int q_{\theta}(\ud \theta') Q^{(U)}_{\theta'}(\ud u')\alpha(\theta,u; \theta',u'). 
$
The  \emph{$\DApr$ correction} of $K$ is defined to be 
\begin{align}\nonumber
K^{\DApr}_{x} (\ud x')
= q_{\theta}(\ud \theta') &Q^{(U)}_{\theta'}(\ud u') \alpha(\theta,u; \theta',u')Q^{(V)}_{\theta' u'}(\ud v') \min\big\{1, \xi'(1)/\xi(1)\big\} \\ \label{dapr-kernel}
&+ [1 - \ga_{\DApr}(x)]\delta_{\theta u v}(\ud \theta',\ud u', \ud v'), 
\end{align}
where $\alpha_{\DApr}(x) = \int q_{\theta}(\ud \theta') Q^{(U)}_{\theta'}(\ud u') \alpha(\theta,u; \theta',u')Q^{(V)}_{\theta' u'}(\ud v') \min\big\{1, \xi'(1)/\xi(1)\big\}$, and $\mathbf{X}\defeq \mathbf{T}\times \mathbf{U}\times \mathbf{V}$, $x\in \mathbf{X}$, $x\defeq(\theta,u,v)$.

Decreasing the variability of $\xi'(1)=\zeta'(1)/\eta'(1)$ by coupling the $u'$ and $v'$ variables can lead to improved mixing of \eqref{dapr-kernel}, and is similar in idea to recently proposed `correlated PM' \cite{deligiannidis-doucet-pitt-kohn} and `MHAAR' \cite{andrieu-doucet-yildirim-chopin} chains.  The mere requirement of reversibility allows the kernel $K$ to be taken to be approximate versions of the two chains listed above, or an approximate DA or `multi-stage DA' \citep[][]{banterle-grazian-lee-robert}.  Regardless, the most straightforward choice for $K$ is the (approximate) PM kernel targeting $\mu$ with proposal $q$, given by,
\begin{align}\label{is-pm-kernel}
K_{\gt u}(\ud \gt', \ud u')
&=
q_{\gt}(\ud\gt')Q_{\gt'}^{(U)}(\ud u')\min\big\{1, \ar{U}{x}{x'}\big\}
\nonumber \\ &\qquad\qquad
+[1-\ga(\gt, u)]\gd_{\gt u}(\ud \gt', \ud u'),
\end{align}
where
$
\ga(\gt,u)
\defeq
\int q_{\gt}(\ud\gt')Q_{\gt'}^{(U)}(\ud u')\min\big\{1, \ar{U}{x}{x'}\big\}.
$

The asymptotic variance of $\DAg$ is never more than that of $\DApr$.
\begin{proposition}\label{prop:das}
  If $K$ is the $\mu$-proposal-rejection kernel \eqref{pr-kernel}, then:
  \begin{enumerate}[(i)]
  \item\label{prop:das-dag} $\Var(K^{\DAg}, g) \le \Var(K^{\DApr}, g)$ for all $g\in L^2(\pi)$.
  \item\label{prop:das-equal} If $V\sim Q_{\theta u}^{(V)}(\cdot)$ with $V=(M, Z^{(1:M)}, \zeta^{(1:M)})$ has the property that 
    \begin{equation}\label{eq:deterministic}
    \zeta(1) = \varphi(\theta,u)
    \end{equation}
is a deterministic function $\varphi$ of $\theta$ and $u$, then $K^\DApr= K^\DAg$.
    \end{enumerate}
\end{proposition}
However, for the reason discussed in Section \ref{sec:application-algorithms}, $\DApr$ is likely more computationally efficient than $\DAg$ in practice.  

We define the PM parent kernel $P$ of $K^{\DAg}$ to be given by
\begin{align}\label{pm-parent-kernel}
P_{\gt u v}(\ud \gt', \ud u', \ud v')
&=
q_{\gt}( \ud \gt') Q^{(U)}_{\gt'}(\ud u')Q^{(V)}_{\gt'u'}(\ud v')
\min\big\{1,\ar{V}{x}{x'}
    \big\}
\nonumber\\ &\qquad\qquad
    + [1-\ga_{\PMP}(\gt,v)]\gd_{\gt u v}(\ud \gt',\ud u',\ud v'),
\end{align}
where
$
\ga_{\PMP}(\gt,v)
\defeq
\int q_{\gt}( \ud \gt') Q^{(U)}_{\gt'}(\ud u')Q^{(V)}_{\gt'u'}(\ud v')
\min\big\{1,\ar{V}{x}{x'}\big\}.
$

We define a probability kernel from $\mathbf{T}$ to $\mathbf{V}$ by
\begin{equation}\label{q-hat}
\hat{Q}_{\gt}^{(V)}(\ud v)
\defeq
\int_{\mathbf{U}} Q_{\gt}^{(U)}(\ud u) Q^{(V)}_{\gt u}(\ud v)
\end{equation}
We then define the following $PM$ kernel with proposal $q$,  
\begin{align}\label{pm-kernel}
M_{\gt v}(\ud \gt', \ud v')
&=
q_{\gt}(\ud\gt') \hat{Q}_{\gt'}^{(V)}(\ud v')\min\big\{1, \ar{V}{x}{x'}\big\}
\nonumber \\ &\qquad\qquad
+[1-\ga_{\PM}(\gt, v)]\gd_{\gt v}(\ud \gt', \ud v'),
\end{align}
targeting $\hat{\pi}(\ud \theta, \ud v)\defeq \int_\mathbf{U} \pi(\ud \theta, \ud u, \ud v)$, where
$
\ga_{\PM}(\gt,v)
\defeq
\int q_{\gt}(\ud\gt')\hat{Q}_{\gt'}^{(V)}(\ud v')\min\big\{1, \ar{V}{x}{x'}\big\}.
$

When $U_k$ and $V_k$ are independent given $\gt$, i.e.
\begin{equation}\label{uncorrelated}
  Q^{(V)}_{\gt u}(\ud v)
  =
  Q^{(V)}_{\gt}(\ud v),
\end{equation}
then $M$ \eqref{pm-kernel} is the standard PM with proposal $q$, since,  
$$
\hat{Q}^{(V)}_{\gt}(\ud v) = Q_{\gt}^{(V)}(\ud v).
$$


\begin{proposition}\label{prop:dag}
  If $K$ is the approximate PM \eqref{is-pm-kernel}, and $L\in \{M,P\}$, then:
  \begin{enumerate}[(i)]
  \item\label{prop:dag-less} $\Var(L,\hat{\zeta}(f)) \le \Var(K^{\DApr}, \hat{\zeta}(f))$ for all $f\in L_\pi^2(\nu)$.
  \item\label{prop:dag-det} If \eqref{eq:deterministic} holds, then
    for all $f \in L_\pi^2(\nu)$,
  $$\Var(L, \hat{\zeta}(f)) \le \Var(K^\DAg,\hat{\zeta}(f)).$$ 
    \end{enumerate}
\end{proposition}

\subsection{Comparison with importance sampling correction}
Note that the following result only involves the weight, not the Dirichlet forms.

\begin{theorem}\label{thm:comparison}
  Suppose Assumption \ref{a:pm-kernels} (PM kernels) holds, and that one of the following conditions for pairs of kernels holds:
  \begin{enumerate}[(I)]
    \item\label{thm:comparison-dapr}
    $L= K^{\DApr}$ is $\DApr$ correction \eqref{dapr-kernel}, and $K$ is $\mu$-proposal-rejection \eqref{pr-kernel}, 
  \item\label{thm:comparison-da}
    $L= K^{\DAg}$ is $\DAg$ correction \eqref{dag-kernel}, and $K$ is $\mu$-reversible,
  \item\label{thm:comparison-pm-parent}
      $L = P$ is the PM parent \eqref{pm-parent-kernel}, and $K$ is the approx. PM \eqref{is-pm-kernel}, or
    \item\label{thm:comparison-pm}
      $L= M$ is the PM kernel \eqref{pm-kernel}, and $K$ is the approx. PM \eqref{is-pm-kernel}.
    \end{enumerate}
Assume $K$ and $L$ are Harris ergodic, and a function $f\in \mathcal{L}^2(\nu)$ is such that $\mathbb{V}_f^{\IS}<\infty$.  
  The following statements hold: 
\begin{enumerate}[(i)]
\item \label{thm:comparison-one}
The IS asymptotic variance \eqref{is-clt} satisfies, with $\lc \defeq \bar{\mu}$-$\essinf \cv$,
\begin{align*}
  \mathbb{V}^{\IS}_f
+ \mu(\mathsf{a}) \Var_{\bar{\mu}}\big( \cv \hat{\zeta}(\bar{f}) \big)
  &\le
    \mu(\mathsf{a}) \norm{w}_{L^\infty(\bar{\mu})} \big\{\Var\big(L, \hat{\zeta}(f) \big) + \Var_\pi\big( \hat{\zeta}(f)\big)\big\}
    + \mathsf{D}_{\bar{f}}
    \\
    \V{\IS}{f}
       + \mu(\mathsf{a}) \Var_{\bar{\mu}}\big( \cv \hat{\zeta}(\bar{f}) \big)
    &\ge 
    \mu(\mathsf{a})\;\cdot\, \lc\;\cdot\; \big\{\Var\big(L, \hat{\zeta}(f) \big) + \Var_\pi\big(\hat{\zeta}(f)\big)\big\}
+ \mathsf{D}_{\bar{f}}.
    \end{align*}
  \item \label{thm:comparison-two}
    With $\mathscr{N}_K\defeq 0$ if $K$ is positive and $\mathscr{N}_K \defeq 1$ if not, the following holds:
  \begin{align*}
    \mathbb{V}^{\IS}_f \le
    &\,\mu(\mathsf{a}) \norm{w^*}_{L^\infty(\mu)} \big\{  \Var\big(L,\hat{\zeta}(f) \big) + \Var_\pi\big(\hat{\zeta}(f)\big)\big\} 
\\
    &+ (1+2\mathscr{N}_K)\mu(\mathsf{a}) \Var_{\bar{\mu}}\big( \cv \hat{\zeta}(\bar{f})\big)
+ \mathsf{D}_{\bar{f}}.
\end{align*}
\end{enumerate}
\end{theorem}
See Remark \ref{thm:peskun-is-remark}\eqref{thm:peskun-is-remark-definitions} for $\cv$ and $\cv^*$.  See Appendix \ref{app:dirichlet} for the proof of Theorem \ref{thm:comparison}, which follows from Theorem \ref{thm:peskun-is}, after bounding the Dirichlet forms.

\section{Discussion and further stability considerations}\label{sec:discussion}
A necessary condition for a successful implementation of an IS or PM scheme is a simple support condition, Assumption \ref{a:pm-kernels}\eqref{a:pm-kernels-support}, that can often be easily ensured by Remark \ref{a:pm-kernels-remark}\eqref{a:pm-kernels-remark-support}.
On the other hand, Theorem \ref{thm:comparison} depends on a uniform bound on the marginal weight $\cv^* \propto \mathbf{m}_1$, with $\mathbf{m}_f(\gt,u)$ as in \eqref{alg:is-scheme-iso-equation}.  This bound is much weaker than a bound on $\cv$, and can often be ensured.   For example, assuming that $\eta(1) \mathbf{m}_1$ is bounded, one can often inflate $\eta(1)$ as in Remark \ref{a:pm-kernels-remark}\eqref{a:pm-kernels-remark-support} to obtain an uniform bound on $\cv^*$.  Other techniques may be applicable if a bounded $\cv^*$ is particularly desired, such as a combination of cutoff functions, approximations, or tempering \citep[see][]{owen-zhou,vihola-helske-franks}. 

When considering a PM/DA implementation, the issue of boundedness of the full weight $\cv\propto \zeta(1)/\eta(1)$ takes particular importance, more so than in the case with IS.  This is because PM and DA are more liable to be poorly mixing, while IS is less affected by noisy estimators.  Namely, if $\zeta(1)$ is not bounded, then PM parent and $K^{\DApr}$, with $K$ as in \eqref{is-pm-kernel}, are not geometrically ergodic (Proposition \ref{prop:not-geo}).

On the other hand, the IS chain may converge fast, even in the case of unbounded $\zeta(1)$.  For example, if $K$ is a random walk $\MH$ chain, then $K$ is geometrically ergodic essentially if $\mu$ has exponential or lighter tails and a certain contour regularity condition holds \cite{jarner-hansen,roberts-tweedie}, where we have said nothing about the exact level estimator $\zeta(1)$.  We then apply Lemma \ref{lem:augmented-kernel}\eqref{lem:augmented-kernel-geometric}, which says that whenever $K$ is geometrically ergodic then so is $\bar{K}$, to conclude that the IS chain is geometrically ergodic, even in the case of unbounded $\zeta(1)$.  This may be beneficial if adaptation is used \citep{andrieu-moulines-ergodicity,andrieu-thoms,roberts-rosenthal-coupling}.

Of course, high variability affects also the IS estimator, but we believe this noise to be a smaller issue in IS, as the noise is in the IS output estimator rather than in the acceptance ratio as in PM/DA.  This can make a significant difference in the evolution and ergodicity of the chains, as described above.


\section*{Acknowledgments}
Support has been provided for JF and MV from the Academy of Finland (grants 274740, 284513 and 312605), and for JF from The Alan Turing Institute.
JF thanks the organisers of the 2017 SMC course and workshop in Uppsala.  

\appendix

\section{Proofs for the Peskun type orderings}\label{app:peskun}

\subsection{Subprobability kernels}\label{app:peskun-subprobability}
Let $K$ be a $\mu$-reversible Markov kernel on $\mathbf{X}$.  For all $\gl \in (0,1]$, $\gl K$ is a \emph{subprobability kernel}: $\gl K(x,\mathbf{X}) \le 1$ for all $x\in\mathbf{X}$.  The \emph{Dirichlet form} $\cE_{\gl K}(f)$ of the subprobability kernel $\gl K$ is 
\begin{equation}\label{eq:dirichlet}
  \cE_{\gl K}(f) \defeq \inner{f}{(1-\gl K)f}_\mu
  =
 \gl \cE_{K}(f) +(1-\gl )\norm{f}_\mu^2,
\end{equation}
defined for $f\in L^2(\mu)$.  For $f\in L^2_0(\mu)$, if $(1-K)\inv f$ exists in $L^2(\mu)$, then by \eqref{def:asvar-eq}, $\Var(K, f) = 2 \inner{ f}{ (1-K)\inv f}_\mu - \mu(f^2)$ \citep[see][]{andrieu-vihola-order}.  Following \cite{andrieu-vihola-order,tierney-note}, we then (formally) extend Definition \ref{def:asvar} of the asymptotic variance to subprobability kernels: for $\gl \in (0,1)$, the operator $(1-\gl K)$ is always invertible, and we define
\begin{equation}\label{eq:asvar-inverse-sub}
\Var(\gl K, f) \defeq 2 \inner{ f}{ (1-\gl K)\inv f}_\mu - \mu(f^2).
\end{equation}
Moreover, \eqref{eq:dirichlet-reversible} and \eqref{eq:dirichlet} imply for $\gl\in (0,1]$ that $1-\gl K$ is a positive operator, i.e.~$\cE_{\gl K}(f)\ge 0$ for all $f\in L^2(\mu)$.  By a result attributed to Bellman \citep[Eq. 14]{bellman}, for positive self-adjoint operators, and used e.g.~in \cite{andrieu-scan,andrieu-vihola-order,caracciolo-pelissetto-sokal,mira-leisen-covariance,mira-geyer-ordering}, we have another asymptotic variance representation: for all $\gl\in (0,1)$ and $f\in L^2_0(\mu)$, 
\begin{equation}\label{eq:var-char}
\Var(\gl K, f) =  2 \sup_{g\in L^2(\mu)} \big\{ 2\inner{f}{g}_\mu - \cE_{\gl K}(g)\big\} -\mu(f^2).
\end{equation}
Here, the supremum is attained with $g\defeq (1-\gl K)\inv f$, in which case \eqref{eq:var-char} simplifies to \eqref{eq:asvar-inverse-sub}.  For $\gl\in (0,1)$, equalities (\ref{eq:asvar-inverse-sub}--\ref{eq:var-char}) hold and are finite for any $f\in L^2_0(\mu)$.  The function $\gl \mapsto \Var( \gl K, f)$ has a limit as $\gl \uparrow 1$ on the extended real numbers $[0,\infty]$, and $\Var(K ,f)$ equals this limit \citep[][]{tierney-note}.

\subsection{Normalised importance sampling ordering}\label{app:peskun-normalised} We set
\begin{equation}\label{left-gap}
  \mathscr{N}_K \defeq
  - 
  \inf_{\mu(g)=0,\mu(g^2)=1} \inner{g}{K g}_\mu
  \end{equation}
for a $\mu$-reversible kernel $K$,
so that the \emph{left spectral gap} of $K$ is $1-\mathscr{N}_K$ \citep[see][]{andrieu-vihola-order}. We have $\mathscr{N}_K \in [-1,1]$ in general, but $\mathscr{N}_K\in[-1,0]$ if $K$ is positive.

The conditions of the next two lemmas will seem more natural once Lemma \ref{lem:helper} is stated.

\begin{lemma}\label{lem:p}
  Suppose $(\mu,\nu, \cv, K, L, \lc ,\uc)$ satisfies Assumption \ref{a:peskun-marginal} on $\mathbf{X}\defeq \mathbf{T}\times\mathbf{Y}$.  Let $\vp\in L^2_0(\nu)$ be such that $\cv \vp \in L^2(\mu)$.  Define $u_\gl \defeq (1-\gl K)\inv( \cv \vp)$ and $\check{u}_\gl \defeq u_\gl - \cv \vp$, in $L^2(\mu)$ for all $\gl\in (0,1)$.  The following hold: 
    \begin{enumerate}[(i)]
    \item\label{lem:p-one} If $u_\gl(\gt,y) = u_\gl(\gt)$, $\gl\in (0,1)$, then
\eqref{thm:peskun-equation} holds.
\item\label{lem:p-two} If $\check{u}_\gl(\gt,y) =\check{u}_\gl(\gt)$, $\gl \in(0,1)$, then \eqref{eq:peskun-augmented} holds, with $\mathscr{N}_K$ as in \eqref{left-gap}. 
    \end{enumerate}
    \end{lemma}
  \begin{proof}
    Note that $L^2(\mu^*)\subset L^2(\nu^*)$ by Assumption \ref{a:peskun-marginal}\eqref{a:peskun-marginal-two}.   For $g\in L^2(\mu^*)$,  
    $$
      \cE_{\gl L}(g) =
      \gl \cE_{L}(g) + (1-\gl) \nu^*(g^2) 
      \le
\uc  \gl \cE_{K}(g)+ (1-\gl) \nu^*(g^2),
      $$
 by Assumption \ref{a:peskun-marginal}\eqref{a:peskun-marginal-one}.  From the above first equality, now for $\gl K$ and $\mu^*$,
      \begin{align}\label{eq:d-bound}
\cE_{\gl L}(g)
&\le
      \uc \big[  \cE_{\gl K}(g) - (1-\gl) \mu^*(g^2)\big] + (1-\gl)\nu^*(g^2)
      \nonumber\\&= \uc  \cE_{\gl K}(g) - (1-\gl) \mu^*\big(g^2[\uc  -\cv^*]\big)
        \le
        \uc  \cE_{\gl K}(g),
      \end{align}
by Assumption \ref{a:peskun-marginal}\eqref{a:peskun-marginal-two}.  Since $1-\gl K$ is self-adjoint on $L^2(\mu)$, we also note that
$$
\cE_{\gl K}(\check{u}_\gl)
=
\cE_{\gl K}(u_\gl - \cv \vp)
=\cE_{\gl K}(u_\gl) + \cE_{\gl K}(\cv \vp) - 2 \norm{\cv \vp}_\mu^2,
$$
as
$
\inner{v_{\gl}}{(1-\gl K) \cv \vp}_\mu = \norm{\cv \vp}_\mu^2.
$
Regardless of $\gl\in (0,1)$, $1-\gl K$ has support of its spectral measure contained in $[0,1+\mathscr{N}_K]$.
Hence, $\cE_{\gl K}( \cv \vp) \le (1+\mathscr{N}_K) \norm{\cv \vp}_{\mu}^2$, so
   \begin{equation}\label{eq:d-simplified}
\cE_{\gl K}(\check{u}_\gl) \le \cE_{\gl K}(u_\gl) + (\mathscr{N}_K-1 )\norm{\cv \vp}_{\mu}^2.
   \end{equation}

   We now compare the asymptotic variances.  By \eqref{eq:asvar-inverse-sub},
\begin{equation*}
  LS
  \defeq \Var(\gl K,\cv \vp) + \norm{\cv \vp}_{\mu}^2
  =
  2\inner{\cv \vp}{u_\gl}_\mu
  =
    2\big[ 2\inner{ \cv \vp}{u_\gl}_{\mu} - \cE_{\gl K}(u_\gl)\big].
\end{equation*}
With $\psi\defeq u_\gl$ for \eqref{lem:p-one}, and with $\psi\defeq \check{u}_\gl$ for \eqref{lem:p-two} using \eqref{eq:d-simplified},
  $$
LS
      \le
      2\big[ 2\inner{\cv \vp}{\psi}_{\mu} - \cE_{\gl K}(\psi)\big] + E_\psi,
$$
      where $E_\psi\defeq 0$ if $\psi=u_\gl$ and $E_\psi\defeq 2(1+\mathscr{N}_K) \norm{\cv \vp}_{\mu}^2$ if $\psi=\check{u}_\gl$.  Hence,
$$
LS
\le
2\big[ 2\inner{  \vp}{\psi}_{\nu} - \cE_{\gl K}(\psi)\big] +E_\psi
\le
2\big[ 2\inner{ \vp}{\psi}_{\nu} - (\uc)\inv \cE_{\gl L}(\psi)\big] + E_\psi,
$$
where we have used \eqref{eq:d-bound}.  Since $\psi\in L^2(\mu^*)\subset L^2(\nu)$, 
\begin{align*}
LS
&\le
\fr{1}{\uc } \Big( 2  \sup_{g\in L^2(\nu)}\big\{ 2\inner{\uc \vp}{g}_{\nu} - \cE_{\gl L}(g)\big\}  - \norm{ \uc \vp}_{\nu}^2\Big) + \uc \norm{\vp}_{\nu}^2 + E_\psi
\\&= 
\uc \big( \Var(\gl L, \vp) + \norm{\vp}_{\nu}^2 \big) + E_\psi,
\end{align*}
by \eqref{eq:var-char}.  We then take the limit $\gl \uparrow 1$ \cite{tierney-note}.   Noting that $\norm{\cv \vp}_\mu^2 = \Var_\mu(\cv \vp)$ since $\mu(\cv \vp) = \nu(\vp)=0$, we conclude.
\end{proof}

  \begin{lemma}\label{lem:p-reverse}
    Suppose the assumptions of Lemma \ref{lem:p} hold, where $\uc$ may be also $\infty$.  If $v_\gl \defeq (1-\gl L)\inv (\vp)$ satisfies $v_\gl (\gt, y) = v_\gl(\gt)$, then \eqref{thm:peskun-equation-reverse} holds.
    \end{lemma}
  \begin{proof}
The lower bound \eqref{thm:peskun-equation-reverse} is trivial if $\lc=0$.  Assume $\lc>0$.  Then $\mu \ll \nu$, $\cv \inv \le \lc\inv$ (implying $L^2(\nu) \subseteq L^2(\mu)$), and $\cE_K(g)\le \lc\inv \cE_L(g)$ for all $g\in L^2(\nu)$. The result follows by applying Lemma \ref{lem:p}\eqref{lem:p-one}.    
    \end{proof}

  \begin{remark} \label{lem:p-remark}
    Regarding Lemma \ref{lem:p} and Lemma \ref{lem:p-reverse}:
    \begin{enumerate}[(i)]
\item The solution $v_\gl$ to the Poisson equation \citep[see][]{meyn-tweedie}, $(1-\gl L)g=\vp$ in $L^2(\mu)$, is also used in \citep[Thm.~17]{andrieu-vihola-order} as a lemma for the proof of the convex order criterion Peskun type ordering for PM chains \citep[Thm.~10]{andrieu-vihola-order}. 
\item\label{lem:p-remark-constant} It is reasonable to use a single constant $\uc$ in Assumptions \ref{a:peskun-marginal}(\ref{a:peskun-marginal-one}--\ref{a:peskun-marginal-two}).  If one replaces Assumption \ref{a:peskun-marginal}\eqref{a:peskun-marginal-two} with $\cv^* \le \uc'$ $\mu^*-a.e.$, then, if $\uc'< \uc$, one obtains the same result after bounding a nonpositive quantity by zero in \eqref{eq:d-bound}.  If $\uc'> \uc$, then one would need to impose the unappealing condition that $\sup_{\gl \in (0,1)} \norm{ u_\gl}_{\mu^*}^2 <\infty$ and add a positive constant involving this bound to the final results.  Anyways, for the the application in this paper, we have $\uc=\uc'$ (Lemma \ref{lem:comparison-helper}).
    \item\label{lem:p-remark-subspace} Assumption \ref{a:peskun-marginal}\eqref{a:peskun-marginal-one} can be replaced with the weaker assumption that $\cE_{L}(g)\le \uc \, \cE_{K}(g)$ for all $g\in \mathcal{G} \subset L^2(\mu^*)$, where $\mathcal{G} \defeq \{u_\gl: \gl\in (0,1)\}$ for \eqref{lem:p-one} and $\mathcal{G} \defeq \{\check{u}_\gl: \gl\in (0,1)\}$ for \eqref{lem:p-two}.
    \end{enumerate}
    \end{remark}

  \begin{lemma}\label{lem:helper}
    Let $K$ be a $\mu$-reversible chain on $\mathbf{X}=\mathbf{T}\times\mathbf{Y}$.  For $h\in L^2(\mu)$ and $\gl\in(0,1)$, set $h_\gl \defeq (1- \gl K)\inv h$ and $\ch \defeq h_\gl - h$, which are in $L^2(\mu)$.
    \begin{enumerate}[(i)]
    \item\label{lem:helper-one}  If $\mathbf{Y} =\{y_0\}$ is the trivial space, then $h_\gl(\gt,y) = h_\gl(\gt).$ 
      \item\label{lem:helper-two} If $K$ is an augmented kernel, then $\check{h}_\gl(\gt,y) = \check{h}_\gl(\gt)$.  Moreover, if also $h(\gt,y) = h(\gt)$, then $h_\gl(\gt,y) = h_\gl(\gt).$  
    \end{enumerate}
  \end{lemma}
  
  \begin{proof} \eqref{lem:helper-one} is clear.  For \eqref{lem:helper-two}, we write the series representation for the inverse of an invertible operator and use Lemma \ref{lem:augmented-kernel}\eqref{lem:augmented-kernel-forget}, to get that,
  $$
  h_\gl(\gt,y) = \sum_{n=0}^\infty \gl^n K^n h(\gt,y) = h(\gt,y) + \sum_{n=1}^\infty \gl^n \dot{K}^n (Q h) (\gt).
  $$
The result then follows.
  \end{proof}

  \begin{proof}[Proof of Theorem \ref{thm:peskun}]
    The upper bound \eqref{thm:peskun-equation} follows from Lemma \ref{lem:p}\eqref{lem:p-one} and Lemma \ref{lem:helper}\eqref{lem:helper-one}, while \eqref{thm:peskun-equation-reverse} follows from Lemma \ref{lem:p-reverse} and Lemma \ref{lem:helper}\eqref{lem:helper-one}. 
  \end{proof}
  
\begin{proof}[Proof of Theorem \ref{thm:peskun-augmented}]
  Follows by Lemma \ref{lem:p} and Lemma \ref{lem:helper}\eqref{lem:helper-two}. 
\end{proof}

\subsection{Importance sampling schemes}\label{app:peskun-is}
The following CLT, based on Proposition \ref{prop:kipnis-varadhan}, and asymptotic variance formula, are \citep[Theorem 7 and 15]{vihola-helske-franks}.
\begin{proposition}\label{prop:is-asvar}
  Under Assumption \ref{a:is-clt}, the IS estimator \eqref{is-estimator} satisfies the CLT \eqref{is-clt}, with limiting variance 
  $
  \mathbb{V}^{\IS}_f
  =
  \mu(\mathsf{a})
  \big[
\Var(K, \mathbf{m}_{f}) + \mu(\mathsf{a} \mathsf{v}_{\bar{f}})
\big]/c_\xi^2.  
  $
  \end{proposition}

\begin{proof}[Proof of Theorem \ref{thm:peskun-is}]
We first note that
  $$
  \xi(f)
  \defeq
  \fr{\zeta(f)}{\eta(1)}
  =
  \frac{c_\zeta}{c_\eta}\cdot
  \frac{c_\eta}{c_\zeta} \frac{\zeta(1)}{\eta(1)}\cdot
  \frac{\zeta(f)}{\zeta(1)}
  =
  c_\xi \cv \hat{\zeta}(f).
  $$ 
By Slutsky's lemma applied to \eqref{is-estimator} in the IS0 case,
$$
  \mathbb{V}^{IS0}_f
  =
  \Var\big(\bar{K}, \xi(f)\big)/c_\xi^2
  =
  \Var\big(\bar{K}, w\hat{\zeta}(f)\big).
  $$
  Then \eqref{thm:peskun-is-one} follows by Theorem \ref{thm:peskun}, and \eqref{thm:peskun-is-two} by Theorem \ref{thm:peskun-augmented}, for the IS0 case.
To prove the result for the ISJ case, we first note the relationship  
  \begin{equation*}\label{eq:is-asvar-relationship}
    \mathbb{V}^{ISJ}_f
    =
    \mu(\ga) c_\xi^{-2}\Big[
      \Var(K, \mathbf{m}_{f}) + \mu(v_{\bar{f}}) + \mu(\ga \tilde{v}_{\bar{f}} - v_{\bar{f}})\Big]
    =
    \mu(\ga) \mathbb{V}^{IS0}_f + \tilde{D}_{\bar{f}},
  \end{equation*}
from Proposition \ref{prop:is-asvar}.  The result then follows from the IS0 case.
  \end{proof}


\section{Proofs for main comparison application}\label{app:dirichlet}
\begin{lemma}\label{lem:dirichlet}
  Let $(K,L)$ be the pair of kernels as in \eqref{thm:comparison-dapr}, \eqref{thm:comparison-da}, or \eqref{thm:comparison-pm-parent} of Theorem \ref{thm:comparison}, where we assume that $(\bar{\mu},\nu,\cv)$ satisfies Assumption \ref{a:is}, with $(\bar{K},\bar{\mu})$ the  $Q^{(V)}$-augmentation of $K$ \eqref{is-kernel}.  Then, the following hold:
  \begin{enumerate}[(i)]
  \item \label{lem:dirichlet-one} If $\norm{\cv}_{L^\infty(\bar{\mu})} <\infty$, then $\cE_L(g) \le \norm{\cv}_{L^\infty(\bar{\mu})} \cE_{ \bar{K}}(g)$ for all $g\in L^2(\bar{\mu})$.
    \newline If $\lc\defeq \bar{\mu}$-$\essinf w$, then $\cE_L(g) \ge \lc\, \cE_{ \bar{K}}(g)$ for all $g\in L^2(\bar{\mu})$. 
    \item \label{lem:dirichlet-two} If $\norm{\cv^*}_{L^\infty(\mu)} <\infty$, then $\cE_L(g) \le \norm{\cv^*}_{L^\infty(\mu)} \cE_{ \bar{K}}(g)$ for all $g\in L^2(\mu)$.
  \end{enumerate}
    \end{lemma}
    
    \begin{proof}
  This is done separately below for the cases $L\in\{P, K^{\DApr},K^{\DAg}\}$.
Set $G\defeq [g(x) - g(x')]^2$, $g\in L^2(\bar{\mu})$, with $x,x'\in \mathbf{X}\defeq \mathbf{T}\times\mathbf{U}\times \mathbf{V}$.  Then, 
    \begin{align*}\allowdisplaybreaks
      \cE_{P}(g)
      &=
      \fr{1}{2} \int \pi(\ud x) q_{\gt}(\ud \gt') Q_{\gt'}^{(U)}(\ud u') Q_{\gt' u'}^{(V)}(\ud v')
      \min\big\{
1, \ar{V}{x}{x'}
      \big\} G
      \\&=
\fr{1}{2} \int \bar{\mu}(\ud x) q_{\gt}(\ud \gt') Q_{\gt'}^{(U)}(\ud u') Q_{\gt' u'}^{(V)}(\ud v')
      \min\big\{
\cv(x), \cv(x) \ar{V}{x}{x'}
\big\} G
\\&=
\fr{1}{2} \int \bar{\mu}(\ud x) q_{\gt}(\ud \gt') Q_{\gt'}^{(U)}(\ud u') Q_{\gt' u'}^{(V)}(\ud v')
      \min\big\{
      \cv(x), \cv(x') \ar{U}{x}{x'}
\big\} G,
    \end{align*}
because $\cv(x) \ar{V}{x}{x'} = \cv(x') \ar{U}{x}{x'}$, well-defined on the set of interest.
    We then use the bounds $\lc\le \cv \le \norm{\cv}_{L^\infty(\bar{\mu})}$ $\bar{\mu}$-a.e.~to conclude \eqref{lem:dirichlet-one} for $L=P$.  

Now assume $g\in L^2(\mu)$, so $G= [g(\gt, u ) - g(\gt', u')]^2$.  By Jensen's inequality and concavity of $(x,x')\mapsto \min\{x,x'\}$ when one of $x,\, x'\ge 0$ is held fixed, 
\begin{align*}
  \cE_{P}(g)
  &=
 \fr{1}{2} \int \bar{\mu}(\ud x)  q_{\gt}(\ud \gt') Q_{\gt'}^{(U)}(\ud u')G \int Q_{\gt' u'}^{(V)}(\ud v') 
      \min\big\{
\cv(x), \cv(x') \ar{U}{x}{x'}
\big\} 
\\&\le
\fr{1}{2} \int \bar{\mu}(\ud x)  q_{\gt}(\ud \gt') Q_{\gt'}^{(U)}(\ud u') G
      \min\big\{
\cv(x), \cv^*(\gt',u') \ar{U}{x}{x'}
\big\}. 
\end{align*}
Here, we have used that $\ar{U}{x}{x'}$ does not depend on $v'\in \mathbf{V}$, and that
$$
\int w(x) Q^{(V)}_{\gt u}(\ud v)
=
\fr{c_\eta}{c_\zeta}\frac{1}{\eta(1)} \int \zeta(1) Q^{(V)}_{\gt u}(\ud v) 
=
\fr{\pi^*(\ud \gt, \ud u)}{\mu(\ud \gt, \ud u)}
=
w^*(\gt, u).
$$
We then apply Jensen again, this time integrating out $v\in\mathbf{V}$, to get,
\begin{align*}
  &\cE_P(g)
  \\&\le
\fr{1}{2} \int \ud \gt Q_{\gt}^{(U)}(\ud u) \fr{\eta(1)}{c_\eta}  q_{\gt}(\ud \gt') Q_{\gt'}^{(U)}(\ud u') G
\int Q_{\gt u}^{(V)}(\ud v)
      \min\big\{
\cv(x), \cv^*(x') \ar{U}{x}{x'}
\big\}
\\&\le
\fr{1}{2} \int \ud \gt Q_{\gt}^{(U)}(\ud u) \fr{\eta(1)}{ c_\eta}  q_{\gt}(\ud \gt') Q_{\gt'}^{(U)}(\ud u') 
      \min\big\{
\cv^*(\gt,u), \cv^*(\gt',u') \ar{U}{x}{x'}
\big\} G.
\end{align*}
We then apply the bound $w^* \le \norm{\cv^*}_{L^\infty(\mu)}$ $\mu$-a.e.~and use the fact that $\cE_K(g) = \cE_{\bar{K}}(g)$ for all $g\in L^2(\mu)$ to conclude \eqref{lem:dirichlet-two} for $L=P$.  

Now consider the case $L=K^{\DApr}$.  With $G\defeq [g(x) - g(x')]^2$ on $\mathbf{X}^2$, 
  \begin{align*}\allowdisplaybreaks
    \cE_{K^{\DApr}}(g)
    &=
    \fr{1}{2} \int \pi(\ud x) q_{\theta}(\ud \theta')Q^{(U)}_{\theta'}(\ud u') \alpha(\theta, u; \theta',u') Q_{\gt' u'}^{(V)}(\ud v')
    \min\Big\{
      1, \frac{\cv(x')}{\cv(x)} 
      \Big\}
      G
      \\
      &=
      \fr{1}{2} \int \bar{\mu}(\ud x) q_{\theta}(\ud \theta')Q^{(U)}_{\theta'}(\ud u') \alpha(\theta, u; \theta',u')  Q_{\gt' u'}^{(V)}(\ud v')
    \min\big\{
      \cv(x) , \cv(x')
      \big\}
      G,
  \end{align*}
  for all $g\in L^2(\bar{\mu})$.  As before, this allows us to conclude \eqref{lem:dirichlet-one} for $L=K^{\DApr}$.  

  Now assume $g\in L^2(\mu)$, with $G\defeq [g(\gt, u ) - g(\gt', u')]^2$.  By Jensen,
  \begin{align*}
    \cE_{K^{\DApr}}(g)
    &\le
        \fr{1}{2} \int \bar{\mu}(\ud x) q_{\theta}(\ud \theta')Q^{(U)}_{\theta'}(\ud u') \alpha(\theta, u; \theta',u')  G 
    \min\big\{
    \cv(x),  \cv^*(\gt',u')
    \big\}
\\&\le
    \fr{1}{2} \int \mu(\ud \gt, \ud u) q_{\theta}(\ud \theta')Q^{(U)}_{\theta'}(\ud u') \alpha(\theta, u; \theta',u')  G 
    \min\big\{
\cv^*(\gt,u), \cv^*(\gt',u')
\big\},
\end{align*}
which allows us to conclude \eqref{lem:dirichlet-two} as before.

Now consider the case $L=K^{\DAg}$.  With $G\defeq [g(x) - g(x')]^2$ on $\mathbf{X}^2$, 
  \begin{align*}\allowdisplaybreaks
    \cE_{K^{\DAg}}(g)
    &=
    \fr{1}{2} \int \pi(\ud x) K_{\gt u}(\ud \gt', \ud u') Q_{\gt' u'}^{(V)}(\ud v')
    \min\Big\{
      1, \frac{\cv(x')}{\cv(x)} 
      \Big\}
      G
      \\
      &=
      \fr{1}{2} \int \bar{\mu}(\ud x) K_{\gt u}(\ud \gt', \ud u') Q_{\gt' u'}^{(V)}(\ud v')
    \min\big\{
      \cv(x) , \cv(x')
      \big\}
      G,
  \end{align*}
  for all $g\in L^2(\bar{\mu})$.  As before, this allows us to conclude \eqref{lem:dirichlet-one} for $L=K^{\DAg}$.  

  Now assume $g\in L^2(\mu)$, with $G\defeq [g(\gt, u ) - g(\gt', u')]^2$.  By Jensen,
  \begin{align*}
    \cE_{K^{\DAg}}(g)
    &\le
        \fr{1}{2} \int \bar{\mu}(\ud x) K_{\gt u }(\ud \gt', \ud u') G 
    \min\big\{
    \cv(x),  \cv^*(\gt',u')
    \big\}
\\&\le
    \fr{1}{2} \int \mu(\ud \gt, \ud u) K_{\gt u }(\ud \gt', \ud u') G 
    \min\big\{
\cv^*(\gt,u), \cv^*(\gt',u')
\big\},
\end{align*}
which allows us to conclude \eqref{lem:dirichlet-two} as before.
    \end{proof}
    
\begin{lemma}\label{lem:comparison-helper}
  With assumptions as in Lemma \ref{lem:dirichlet}, and additionally assuming that $K$ and $L$ determine Harris ergodic chains, the following hold:
  \begin{enumerate}[(i)]
  \item \label{prop:comparison-helper-one} If $\norm{w}_{L^\infty(\bar{\mu})} <\infty$, then $(\bar{\mu}, \pi ,\cv, \bar{K}, L, \lc, \norm{w}_{L^\infty(\bar{\mu})} )$ satisfies Assumption \ref{a:peskun}.
  \item \label{prop:comparison-helper-two} If $\norm{w^*}_{L^\infty(\mu)} <\infty$, then $(\bar{\mu}, \pi , \cv , \bar{K}, L, 0, \norm{ w^* }_{L^\infty(\mu)})$ satisfies Assumption \ref{a:peskun-marginal}. 
    \end{enumerate}
  \end{lemma}

\begin{proof}
 Lemma \ref{lem:dirichlet}\eqref{lem:dirichlet-one} and \eqref{lem:dirichlet-two} imply respectively \eqref{prop:comparison-helper-one} and \eqref{prop:comparison-helper-two}. 
\end{proof}
    
\begin{proof}[Proof of Theorem \ref{thm:comparison}]
The support condition Assumption \ref{a:pm-kernels}\eqref{a:pm-kernels-support} implies that $(\bar{\mu},\pi, \cv)$ satisfies Assumption \ref{a:is}.
Under conditions \eqref{thm:comparison-dapr}, \eqref{thm:comparison-da}, or  \eqref{thm:comparison-pm-parent}, the result follows by Lemma \ref{lem:comparison-helper} and Theorem \ref{thm:peskun-is}.

Assume condition \eqref{thm:comparison-pm}.  Because $g\defeq \hat{\zeta}(f)$ is a function on $\mathbf{X}=\mathbf{T}\times\mathbf{U}\times\mathbf{V}$ which does not depend on the second coordinate, $P^k g(\gt, u, v) = M^k g(\gt,v)$ for all $(\gt,u,v)\in \mathbf{X}$ and $k\ge 1$.  Therefore, $\Var(M,g)= \Var(P,g)$.
\end{proof}

\begin{proof}[Proof of Proposition \ref{prop:das}]
For any $g\in L^2(\pi)$, set $G\defeq [g(\theta,u,v) - g(\theta',u',v')]^2$.  We have
\begin{equation*}
  \cE_{K^\DAg}(g)
  =
  \cE_{K^\DApr}(g) +
  \frac{1}{2}\int \pi(\ud x) r_K(\theta,u) Q_{\theta u}^{(V)}(\ud v') \min\bigg\{1, \frac{\xi'(1)}{\xi(1)} \bigg\}G,
\end{equation*}
so \eqref{prop:das-dag} follows from the covariance ordering.
For \eqref{prop:das-equal}, we have
$$
K_x^{\DAg}(\ud x') = K_x^{\DApr}(\ud x') - (1-\alpha_{\DApr}(x)) \delta_x(\ud x')
+ r_K(\theta,u)\delta_x(\ud x') +(1-\alpha_{\DAg}(x))\delta_x(\ud x'),
$$
from which we conclude, since
$\alpha_{\DAg}(x) = \alpha_{\DApr}(x) + r_K(\theta,u).$
\end{proof}

\begin{proof}[Proof of Proposition \ref{prop:dag}]
  \eqref{prop:dag-less} is essentially well-known \citep[see][]{banterle-grazian-lee-robert}, and \eqref{prop:dag-det} is straightforward to prove.
  \end{proof}


\section{Properties of augmented kernels and ergodicity}\label{app:misc}  
For measurable functions $V:\mathbf{X}\ra [1,\infty)$ and $f:\mathbf{X}\ra \mathbb{R}$, we set
  $$
  \norm{\nu}_V \defeq \sup_{f:|f|\le V} \nu (f),
  \qquad
  \text{and}
  \qquad
  \norm{f}_V \defeq \sup_{x\in \mathbf{X}} \fr{|f(x)|}{V(x)}
  $$
  for any finite signed measure $\nu$ on $\mathbf{X}$.
\begin{definition}
  A $\mu$-invariant Markov chain $K$ on $\mathbf{X}$ is said to be
  \begin{enumerate}[(i)]
  \item \emph{$V$-geometrically ergodic} if there is a function $V:\mathbf{X}\ra [1,\infty)$ such that
    $$
\norm{K^{n}(x,\cdot) - \mu(\cdot)}_V \le R V(x) \rho^n
$$
for all $n\ge 1$, where $R<\infty$ and $\rho\in (0,1)$ are constants.
\item \emph{uniformly ergodic} if $K$ is $1$-geometrically ergodic.
    \end{enumerate}
  \end{definition}

  \begin{lemma}\label{lem:augmented-kernel}
    Let $K_{\gt y}(\ud \gt', \ud y')  = \dot{K}_{\gt}(\ud \gt')Q_{\gt'}(\ud y')$ be an augmented kernel on $\mathbf{T}\times\mathbf{Y}$.  
    \begin{enumerate}[(i)]
    \item \label{lem:augmented-kernel-measures}
      The invariant measures of $K$ and $\dot{K}$ satisfy $(\mu K = \mu \Longrightarrow \mu^* \dot{K}=\mu^*)$, and
      $
      (\dot{\mu} \dot{K}=\dot{\mu} \Longrightarrow \mu K=\mu),
      \text{ where }
      \mu(\ud \gt, \ud y) \defeq \dot{\mu}(\ud \gt) Q_{\gt}(\ud y).
      $
      These implications hold with invariance replaced with reversibility.
\item \label{lem:augmented-kernel-harris} $K$ is $\mu$-Harris ergodic $\iff$ $\dot{K}$ is $\dot{\mu}$-Harris ergodic.
\item \label{lem:augmented-kernel-forget} For all $f\in L^1(\mu)$ and $n\ge 1$,
      $
K^n f(\gt,y) = \dot{K}^n(Qf)(\gt).
$      
\item \label{lem:augmented-kernel-aperiodic}
  $K$ is aperiodic $\iff$ $\dot{K}$ is aperiodic.  $K$ is positive $\iff$ $\dot{K}$ is positive.
\item\label{lem:augmented-kernel-geometric} $K$ is geometrically ergodic $\iff$ $\dot{K}$ is geometrically ergodic.
\item\label{lem:augmented-kernel-uniform} $K$ is uniformly ergodic $\iff$ $\dot{K}$ is uniformly ergodic.
      \end{enumerate}
  \end{lemma}
\begin{proof}
  (\ref{lem:augmented-kernel-measures}--\ref{lem:augmented-kernel-forget}) are \citep[Lem.~21]{vihola-helske-franks}.  Proof of \eqref{lem:augmented-kernel-aperiodic} is straightforward.

  For \eqref{lem:augmented-kernel-geometric}, consider first the case that $\dot{K}$ is $\dot{V}$-geometrically ergodic:
  $$
  \sup_{|f|\le \dot{V}}\abs{ \dot{K}^n(f)(\gt) - \dot{\mu}(f)}\le R \dot{V}(\gt) \rho^n,
  \qquad
  n\ge 1,
$$
with $\dot{V}:\mathbf{T}\ra [1,\infty)$ and constants $R$ and $\rho$.  Define $V(\gt, y) \defeq \dot{V}(\gt)$.  By \eqref{lem:augmented-kernel-forget},
  \begin{equation}\label{lem:augmented-kernel-geometric-Qf}
  \sup_{|f|\le V}\abs{K^n f(\gt,y)  - \mu(f)}
  =
  \sup_{|f|\le V} \abs{ \dot{K}^n (Qf)(\gt) - \dot{\mu}(Qf)}.
  \end{equation}
  Since $Qf(\gt, y) \le QV(\gt,y)= \dot{V}(\gt)$, we get that $K$ is $V$-geometrically ergodic.

  Assume now that $K$ is $V$-geometrically ergodic.  Using \eqref{lem:augmented-kernel-geometric-Qf}, we have,
  \begin{equation}\label{lem:augmented-kernel-geometric-g}
  \sup_{|f|\le V}\abs{K^n f(\gt,y)  - \mu(f)}
  =
  \sup_{g=Qf:|f|\le V} \abs{\dot{K}^n g(\theta) -\dot{\mu}(g)},
  \end{equation}
  for $n\ge 1$.  Define $\dot{V}(\gt)\defeq \inf_{y} V(\gt, y).$  For all $g$ such that $|g(\gt)|\le \dot{V}(\gt)$, set $f(\gt,y)\defeq g(\gt)$.  Then $|f|\le V$ and $Qf=g$.  By \eqref{lem:augmented-kernel-geometric-g}, $\dot{K}$ is $\dot{V}$-geometrically ergodic.  This proves \eqref{lem:augmented-kernel-geometric}, and \eqref{lem:augmented-kernel-uniform} follows from the form of $\dot{V}$ and $V$.
\end{proof}

\begin{proposition}\label{prop:not-geo}
Consider the PM parent kernel \eqref{pm-parent-kernel} and the $\DApr$ kernel $K^{\DApr}$ \eqref{dapr-kernel}, with $K$ as in \eqref{is-pm-kernel}.  If $\zeta(1)$ is not bounded, then PM parent and $K^{\DApr}$ are not $V$-geometrically ergodic.
  \end{proposition}

\begin{proof}
This is \cite[Thm.~8]{andrieu-roberts} for PM chains.  To prove that result for PM chains, or in particular for the PM parent chain \eqref{pm-parent-kernel}, \cite{andrieu-roberts} show that for all $\epsilon>0$, 
\begin{equation}\label{not-geo-epsilon}
  \nu\big( \mathbf{1}\{ \ga_{\PMP} \le \epsilon\}\big)>0.
\end{equation}
By \citep[Thm.~5.1]{roberts-tweedie}, one concludes that the PM parent is not $V$-geometrically ergodic \cite{andrieu-roberts}.  Moreover, from 
\begin{equation}\label{minmin}
\min\{
1, \ar{U}{x}{x'}
\}
\min\{
1, \cv(x')/\cv(x)
\}
\le
\min\{1, \ar{V}{x}{x'}
\},
\end{equation}
it follows that $\ga_{\DApr}(x)\le \ga_{\PMP}(x)$.  By \eqref{not-geo-epsilon}, one concludes that $K^{\DApr}$ also is not $V$-geometrically ergodic.  
\end{proof}

    \section{Toy examples of two extremes}\label{app:examples}
        \begin{figure}[t]
          \caption{Two versions and two behaviours}\label{fig:regions-is-da-better}
  \centering
\begin{subfigure}[b]{.5 \textwidth}\label{fig:region-da-better}
  $$
\begin{tikzpicture}[scale=.6] 
  \draw (0,0) ellipse(5cm and 2cm);
  \filldraw [lightgray] (-2,0) circle  (1cm);
   \path (2,0) node(a)  {$2$}
   (-2.5,0) node(b) [] {$0$}
   (-1.5,0) node(c) [] {$1$}
   (-2.2,-.6) node(d) [] {$\nu^*$}
      (0,-1) node(e) [] {$\mu^*$};
   \draw (-2,-1) -- (-2,1);
\end{tikzpicture}
$$
\caption{`PM/DA better' case}
  \end{subfigure}%
\begin{subfigure}[b]{.5 \textwidth}\label{fig:region-is-better}
$$
\begin{tikzpicture}[scale=.6] 
  \draw (0,0) ellipse(5cm and 2cm);
  \filldraw [lightgray] (2,0) circle  (1cm);
  \filldraw [lightgray] (-2,0) circle (1cm);
  \path (0,0) node(a)  {$1$}
  (2,0) node(b) [] {$2$}
  (-2,0) node(c) [] {$0$}
     (-2.2,-.6) node(d) [] {$\nu^*$}
      (0,-1) node(e) [] {$\mu^*$}
     (1.8,-.6) node(d) [] {$\nu^*$};
\end{tikzpicture}
$$
\caption{`IS better' case}
  \end{subfigure}
  \end{figure}    
        Let $\mathbf{X}\defeq \{0,1,2\}$ and consider the two mass allocations for probabilities $\mu$ and $\nu$ on $\mathbf{X}$ and function $f\in L^2_0(\nu)$ given pictorially in Figure \ref{fig:regions-is-da-better} and precisely in Figure \ref{fig:mass-is-da-better}.
        \begin{figure}[!tbh]
          \caption{Mass allocations for $\mu$, $\nu$, and $f$ on $\mathbf{X}=\{0,1,2\}$, $a\in [\fr{1}{2},1)$.}\centering%
    \label{fig:mass-is-da-better}%
    \begin{subfigure}[b]{.5 \textwidth}%
    $$
\begin{matrix}
\mu&=&&(&\fr{1-a}{2}&\fr{1-a}{2}&a&)\\
\nu&=&&(&{1}/{2}&{1}/{2}&0&)\\
f&=&\vphantom{\fr{\sqrt{2}}{\sqrt{a + a^2}}}&(&1&-1&0&)
\end{matrix}
$$
\caption{`MH/DA better' case}
 \label{fig:mass-da-better}
    \end{subfigure}%
    \begin{subfigure}[b]{.5 \textwidth}
    $$
\begin{matrix}
  \mu&=&&(&1/3&1/3&1/3&)\\
\nu&=&&(&\fr{a}{2}&\fr{1-a}{2}&{1}/{2}&)\\
f&=&\fr{\sqrt{2}}{\sqrt{a + a^2}} &(&1&0&-a&)
\end{matrix}
$$
\caption{`IS better' case}
      \label{fig:mass-is-better}
    \end{subfigure}%
    \end{figure}
        Denote by $q^{(r)}$ the (reflected) random walk proposal on $\mathbf{X}$, given by $q^{(r)}_0(x) = \gd_1(x)$, $q^{(r)}_1(x) = \fr{1}{2}[ \gd_0(x) + \gd_2(x)]$, and $q^{(r)}_2(x) = \gd_1(x)$, and by $q^{(u)}_x(x')$ the uniform proposal on $\mathbf{X}$.  We set $K\defeq$MH$(q\ra \mu)$ and let $L$ be the MH or $\DApr$ kernels, using proposals $q^{(r)}$ or $q^{(u)}$, and targeting $\nu$.  We use a parameter $a\in [\fr{1}{2},1)$ to allow for continuous intensity shifts in the mass allocations in our examples.  Because $\mu$ is constant on the support of $\nu$, one can check that the MH and $\DApr$ kernels coincide for $a\in[\fr{1}{2},1)$.

\begin{table}[]
  \caption{Asymptotic variance as a function of $a\in [1/2,1)$}
    \label{table:asvar}
    \small
    \begin{tabular}{l ccc ccc}
      \toprule
     Proposal &  \multicolumn{1}{c}{$\Var(L,f)$} & $\le$ & $\Var(K,\cv f)$ & $\Var(L,f)$ & $\ge$& $\Var(K,\cv f)$\\
      \toprule
      RW $q^{(r)}$& $1$   && $\fr{1}{1-a}$   & $\fr{-1+8a+a^2}{a^2-1}$ && $\fr{9a}{1+a} $  \\
      \midrule
       uniform $q^{(u)}$& $2 $  && $\fr{1}{1-a}$  & $\fr{-1+10 a-a^2}{(1+a)^2} $ && $\fr{15 a}{ 4(1+a)}$
\\ \bottomrule
    \end{tabular}
\end{table}
\begin{figure}[!htb]
  \begin{tabular}{cc}
    $\underline{\Var(L,f) \le \Var(K,\cv f)}$ &   $\underline{\Var(L,f) \ge \Var(K,\cv f)}$
    \\
  \includegraphics[scale=.5]{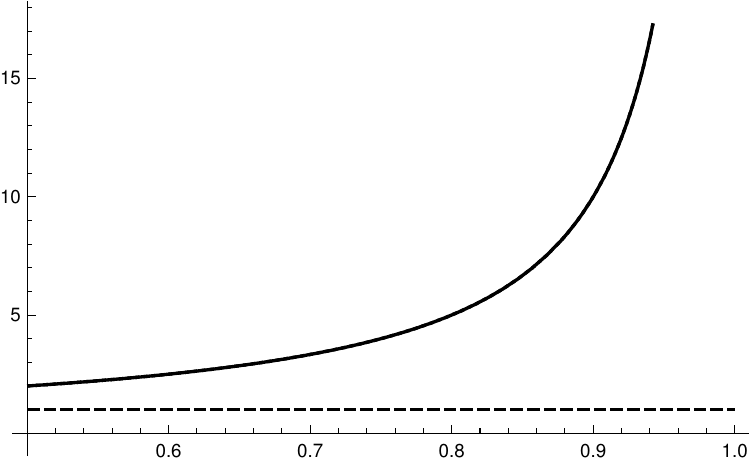}
&
  \includegraphics[scale=.5]{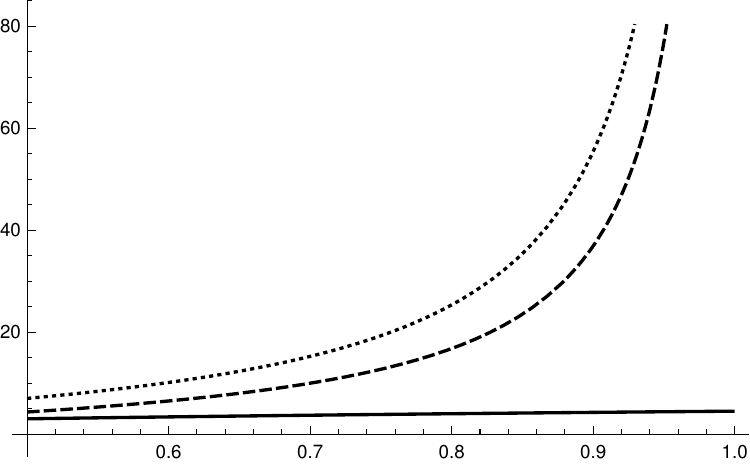}
\\
  \includegraphics[scale=.5]{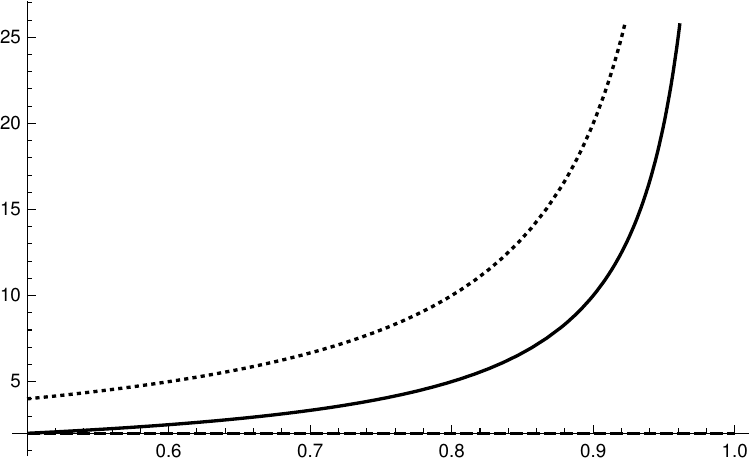}
&
  \includegraphics[scale=.5]{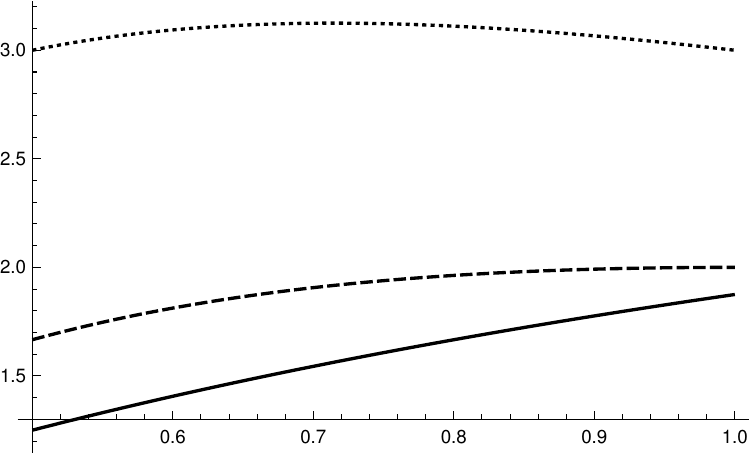}
     \end{tabular}
\caption{Plots from Table \ref{table:asvar}: $\Var(K,\cv f)$ `---', $\Var(L,f)$ `$--$', and $\mathrm{UB}_a(f)$ `$\cdots$', vs. $a\in [\fr{1}{2},1)$.  Here, in the top left, $\mathrm{UB}_a(f)$ exactly coincides with $\Var(K,\cv f)$.}
\label{fig:plots}
\end{figure}

      The resulting IS and MH/DA asymptotic variances, $\Var(K,\cv f)$ and $\Var(L,f)$, can be computed by linear algebra using \citep[see][Cor.~1.5]{kipnis-varadhan}.  They are listed in Table \ref{table:asvar}, and plotted in Figure \ref{fig:plots}.  Here,
\begin{equation}
\mathrm{UB}_a(f) \defeq \max(\cv) \Var(L,f) + \nu(f^2[\max(\cv) - \cv]).
\end{equation}
is the upper bound on $\Var(K,\cv f)$ from Corollary \ref{cor:intro}.




\bibliographystyle{abbrv}  


\end{document}